\documentclass[12pt]{article}
\usepackage[lmargin=1in,rmargin=1in,tmargin=1in,bmargin=1in]{geometry}

\usepackage{times}
\usepackage{soul}
\usepackage{url}
\usepackage[hidelinks]{hyperref}
\usepackage[utf8]{inputenc}
\usepackage[small]{caption}
\usepackage{graphicx}
\usepackage{amsmath}
\usepackage{amsthm}
\usepackage{booktabs}
\usepackage[utf8]{inputenc}
\usepackage{pifont}
\newcommand{\cmark}{\ding{51}}%
\newcommand{\xmark}{\ding{55}}%
\usepackage{booktabs}
\usepackage{nicefrac}

\usepackage[usenames,svgnames,table]{xcolor}
\usepackage{hyperref}[backref=page] 
\hypersetup{
colorlinks   = true,
linkcolor    = Red, 
urlcolor     = Blue, 
citecolor    = Blue 
}

\usepackage[capitalise,noabbrev]{cleveref}
\Crefname{remark}{Remark}{Remarks}
\Crefname{rmk}{Remark}{Remarks}
\Crefname{dfn}{Definition}{Definitions}
\Crefname{thm}{Theorem}{Theorems}
\Crefname{cor}{Corollary}{Corollaries}
\Crefname{lem}{Lemma}{Lemmas}
\Crefname{examplex}{Example}{Examples}
\Crefname{prop}{Proposition}{Propositions}

\usepackage{setspace}
\usepackage{amsthm}
\usepackage{amsfonts}
\usepackage{amssymb}
\usepackage{authblk}
\usepackage{enumitem}
\usepackage{etoolbox}
\usepackage{mathtools}
\usepackage{multirow}
\usepackage{tabularx}
\usepackage{thmtools}
\usepackage{thm-restate}
\usepackage[textsize=footnotesize,color=red!25,bordercolor=red]{todonotes}
\usepackage{wrapfig}
\usepackage{bigdelim}

\usepackage{tikz}
\usetikzlibrary{fit,calc, shapes.misc, positioning}
\usepgflibrary{arrows}

\colorlet{mygray}{gray!40}

\usepackage[linesnumbered,lined,boxed,ruled,vlined]{algorithm2e} 

\SetKwInput{Parameters}{Parameters}
\SetKwComment{Comment}{$\triangleright$\ }{}
\SetAlFnt{\small}
\SetAlCapFnt{\small}
\SetAlCapNameFnt{\small}
\SetAlCapHSkip{0pt}
\IncMargin{-\parindent}

\makeatletter
\renewcommand{\paragraph}{%
 \@startsection{paragraph}{4}%
 {\z@}{1.1ex \@plus 1ex \@minus .2ex}{-1em}%
 {\normalfont\normalsize\bfseries}%
}
\let\oldnl\nl
\newcommand{\nonl}{\renewcommand{\nl}{\let\nl\oldnl}}
\makeatother

\declaretheorem[name=Lemma]{lem}

\declaretheorem[name=Remark]{remark}
\newtheorem{example}{Example}

\newcommand{\ml}{{\mathcal L}}

\newtheorem{examplex}{Example}
\renewenvironment{example}{\pushQED{\qed}\examplex}{\popQED\endexamplex}
\renewcommand{\>}{{\succ}}
\newcommand{\HRTC}{\textup{\textsc{Hypergraph Rainbow $3$-Colorability}}}
\newcommand{\id}{\rhd}
\newcommand{\MMS}{\textrm{\textup{MMS}}}

\newcommand{\NP}{\textrm{\textup{NP}}}
\newcommand{\NPc}{\textrm{\textup{NP-c}}}
\newcommand{\NPC}{\textrm{\textup{NP-complete}}}

\newcommand{\NPH}{\textrm{\textup{NP-hard}}}
\newcommand{\EF}[1]{\ifstrempty{#1}{\textrm{\textup{EF}}}{\textrm{\textup{EF{$#1$}}}}}
\newcommand{\EFX}{\textrm{\textup{EFX}}}
\newcommand{\EFXc}{\textrm{\textup{EFX-c}}}
\newcommand{\EFXg}{\textrm{\textup{EFX-g}}}

\newcommand{\RM}{\textrm{\textup{RM}}}

\newcommand{\PO}{\textup{PO}}

\newcommand{\Polytime}{\textup{P}}

\newcommand{\ThreeSAT}{\textup{\textsc{3-SAT}}}
\newcommand{\TTTSAT}{\textup{\textsc{(2/2/3)-SAT}}}

\colorlet{mycolor}{blue!25}
\colorlet{opencolor}{black}

\usepackage{natbib}
\setcitestyle{square,aysep={},yysep={;}}
\let\citeauthor\citet
\let\cite\citep

\title{Fairly Dividing Mixtures of Goods and Chores\\ under Lexicographic Preferences}

\author[1]{Hadi Hosseini}
\author[2]{Sujoy Sikdar}
\author[3]{Rohit Vaish}
\author[4]{Lirong Xia}
\affil[1]{Pennsylvania State University\\
	{\small\texttt{hadi@psu.edu}}}
\affil[2]{Binghamton University\\
	{\small\texttt{ssikdar@binghamton.edu}}}
\affil[3]{Indian Institute of Technology Delhi\\
	{\small\texttt{rvaish@iitd.ac.in}}}
\affil[4]{Rensselaer Polytechnic Institute\\
	{\small\texttt{xial@cs.rpi.edu}}}

\begin{document}

\maketitle

\begin{abstract}
We study fair allocation of indivisible goods and chores among agents with \emph{lexicographic} preferences---a subclass of additive valuations. In sharp contrast to the goods-only setting, we show that an allocation satisfying \emph{envy-freeness up to any item} (\EFX{}) could fail to exist for a mixture of \emph{objective} goods and chores. 
To our knowledge, this negative result provides the \emph{first} counterexample for \EFX{} over (any subdomain of) additive valuations. To complement this non-existence result, we identify a class of instances with (possibly subjective) mixed items where an \EFX{} and Pareto optimal allocation always exists and can be efficiently computed. When the fairness requirement is relaxed to \emph{maximin share} (\MMS{}), we show positive existence and computation for \emph{any} mixed instance. More broadly, our work examines the existence and computation of fair and efficient allocations both for mixed items as well as chores-only instances, and highlights the additional difficulty of these problems vis-{\`a}-vis their goods-only counterparts.
\end{abstract}

\section{Introduction}

Fair division of indivisible items encompasses a wide array of real-world applications ranging from inheritance division~\cite{BT96fair}, allocation of public housing units~\cite{BCH+20price} and course allocation~\cite{B11combinatorial}, to 
distributing medical equipment and human resources~\cite{pathak2021fair,babus2020optimal,aziz2021efficient}. These applications often require dealing with resources that can simultaneously be seen as \emph{goods} by some agents---generating positive utility---and \emph{chores} by others---generating negative utility. For example,  medical supplies such as vaccines or ventilators~\cite{pathak2021fair} may result in negative utilities for some regions due to storage or maintenance costs, while being generally seen as positively valued resources by others.

A standard solution concept in the study of fairness is \emph{envy-freeness}~\cite{george1958puzzle,foley1967resource}, which requires that no agent prefers another agent's allocation to its own. 
With indivisible (or discrete) resources, envy-freeness cannot always be guaranteed, motivating the study of its 
relaxations. 
A well-studied, and arguably most desirable, relaxation is \emph{envy-freeness up to any item} (\EFX{}),  which states that any pairwise envy can be eliminated by removing \emph{any} item considered as a good from the envied agent's bundle and \emph{any} item seen as a chore from the envious agent's bundle~\cite{CKM+19unreasonable,ACI+19fair}.

For goods-only problems with additive valuations, the existence and computation of an \EFX{} allocation---except for a few special cases~\cite{PR20almost,M21extension,CGM20efx}---remains open. Moreover, \EFX{} in known to be incompatible with well-studied notions of economic efficiency such as \emph{Pareto optimality}~(\PO{}). 
For chores-only problems or those involving a mixture of goods and chores, little is known about the existence and computation of \EFX{}. Complicating matters further, many of the axiomatic and algorithmic techniques from the goods-only setting do not carry over to mixed items~\cite{bogomolnaia2017competitive,ACI+19fair}.

One plausible approach for tackling such challenging problems is \emph{domain restriction}. This approach has been widely adopted in the computational social choice literature \cite{ELP16preference,FLS18complexity,HL19multiple} to investigate structural and computational boundaries of collective decision-making. 
In this vein, we focus on \emph{lexicographic preferences} as a subdomain of additive preferences to study existential and computational questions regarding \EFX{} allocations. 
Lexicographic preferences provide a succinct language for representing complex preferences~\cite{saban2014note,LMX18voting}, and have been widely-studied in psychology~\cite{GG96reasoning}, machine learning~\cite{SM06complexity}, and social choice theory~\cite{T70problem}.

Restricting preferences to the lexicographic domain has already proven fruitful for goods-only instances. Indeed, \citeauthor{HSV+21fair} have studied the goods problem and showed that under lexicographic preferences, an \EFX{} and \PO{} allocation always exists, can be computed efficiently, and can be further characterized alongside strategyproofness and other desirable properties.
Despite these positive results, the chores-only and the mixed-item problems have largely remained unexplored due to several additional challenges that we describe next.

\paragraph{Mixed Items:}
The allocation of mixed items differs drastically from its goods-only counterpart. 
First, for indivisible goods, several well-studied variations of \emph{picking sequences}~\cite{IJCAI113095,Aziz15:Possible,HL19multiple,beynier2019efficiency} satisfy \EFX{} under lexicographic preferences \cite{HSV+21fair}.
However, for mixed items, these variants may violate \EFX{} even for two agents with lexicographic preferences, as we illustrate in \cref{example:mixed_items}.
Second, for goods-only instances, a picking sequence implies Pareto optimality under lexicographic preferences \cite{HSV+21fair}. In contrast, when dealing with mixed items, sequencibility does not imply Pareto optimality, which leads to additional challenges in designing algorithms for mixed items.

\begin{example}\label{example:mixed_items}\em
Suppose there are two agents $1,2$ and four items $o_1,\dots,o_4$. Each agent's importance ordering~$\id$ over the individual items is as shown below:
\begin{align*}
1: ~&~ \underline{o_1^{-}}  \ \id \ \underline{o_2^+} \ \id \ o^+_3 \ \id \ o^+_{4}\\ \nonumber
2: ~&~ {o^+_2} \ \id \ o^{-}_{1} \ \id \ \underline{o^+_3} \ \id \ \underline{o_{4}^{-}}  
\end{align*}
The superscripts $+$ or $-$ denote whether the agent considers the item to be a good or a chore, respectively. Thus, the item $o_4$ is a good for agent $1$ but a chore for agent $2$, while $o_2$ and $o_3$ are ``common goods'' and $o_1$ is a ``common chore''. Thus, the instance contains \emph{subjective} mixed items because $o_4$ is a good for agent 1 but a chore for agent 2.

An agent's preference over bundles of items is given by the lexicographic extension of its importance ordering $\id$ as follows:
Agent $1$ prefers any bundle that does not contain the chore $o_1$ (including the empty bundle) to any bundle that does, subject to which it prefers any bundle containing the good $o_2$ to any bundle that does not, and so on. Similarly, agent $2$ prefers any bundle containing the good $o_2$ to any other bundle that does not, subject to which any bundle without the chore $o_1$ is preferred over any bundle with it, and so on.

Consider a picking sequence $1221$ wherein agent $1$ picks its favorite item first, followed by back-to-back turns 
for agent $2$ to pick its favorite remaining item, before agent $1$ 
picks the leftover item. The allocation induced by this sequence is underlined: First, agent $1$ picks $o_2$ (its favorite good), followed by agent $2$ picking $o_3$ (its favorite remaining good) and then $o_4$ (the chore it dislikes less between $o_1$ and $o_4$), and finally agent $1$ is left to pick its most-disliked chore $o_1$.

It is easy to verify that this allocation is neither \EFX{} nor Pareto optimal. Indeed, agent $2$ continues to envy agent $1$ even after the perceived chore $o_4$ is removed from its own bundle. Moreover, the above allocation is Pareto dominated by an allocation that gives all items to agent $2$.
\end{example}

\begin{table*}[t]
\centering
\scriptsize
\begin{tabular}{|cl|cc|ll|ll|l|}
 \hline
 \multicolumn{2}{|c|}{\multirow{2}{*}{\textbf{Guarantee(s)}}} & \multicolumn{2}{c|}{\textbf{Goods}} & \multicolumn{2}{c|}{\textbf{Chores}} & \multicolumn{2}{c|}{\textbf{Mixed Items}} & \multicolumn{1}{c|}{\multirow{2}{*}{\textbf{Special Cases}}}\\
 \cline{3-8}
 & & existence & computation & existence & computation & existence & computation &\\
  \hline
  & \EF{} & \xmark{} & \Polytime{}$^\dagger$ & \cellcolor{mycolor}\xmark{} & \cellcolor{mycolor}\NPc{} (Thm.~\ref{thm:EF_Chores_NPC}) & \xmark{} & \cellcolor{mycolor}\NPc{} (Thm.~\ref{thm:EF_Chores_NPC}) & \\
  & \EFX{} & \cmark{} & \Polytime{}$^\dagger$ & \cellcolor{mycolor}\cmark{} & \cellcolor{mycolor}\Polytime{} (Prop.~\ref{prop:EFX+PO_Chores}) & \cellcolor{mycolor}\xmark~(Thm.~\ref{thm:EFX_counterexample}) & {\color{opencolor} Open} & \cellcolor{mycolor}\cmark{}  Thm.~\ref{thm:EFX_Mixed_TopItem}; Cor. \ref{cor:EFX_Mixed_special}\\
  & \EF{1} & \cmark{} & \Polytime{}$^\mathsection$ & \cmark{} & \Polytime{}$^\mathsection$ & \cmark & \Polytime{}$^\mathsection$ &\\
  & \MMS{} & \cmark{} & \Polytime{}$^\dagger$ & \cellcolor{mycolor}\cmark & \cellcolor{mycolor}\Polytime{} (Prop.~\ref{prop:MMS+PO_Chores}) & \cellcolor{mycolor}\cmark & \cellcolor{mycolor}\Polytime{} (Thm.~\ref{thm:MMS_PO_Polytime_Mixed_Items}) & \\
  \hline
  %
  %
  \ldelim\{{4}{18pt}[\PO{} +] & \EF{} & \xmark{} & \Polytime{}$^\dagger$ & \cellcolor{mycolor}\xmark{} & \cellcolor{mycolor}\NPc{} (Cor.~\ref{cor:EF+PO_Chores_NPC}) & \xmark{} & \cellcolor{mycolor}\NPc{} (Cor.~\ref{cor:EF+PO_Chores_NPC}) & \\
  & \EFX{} & \cmark{} & \Polytime{}$^\dagger$ & \cellcolor{mycolor}\cmark{} & \cellcolor{mycolor}\Polytime{} (Prop.~\ref{prop:EFX+PO_Chores}) & \cellcolor{mycolor}\xmark~(Thm.~\ref{thm:EFX_counterexample}) & {\color{opencolor} Open} & \cellcolor{mycolor} \\
  & \EF{1} & \cmark{} & \Polytime{}$^\mathsection$ & \cmark{} & \Polytime{}$^\mathsection$ & {\color{opencolor} Open} & {\color{opencolor} Open} & \cellcolor{mycolor}\\
  & \MMS{} & \cmark{} & \Polytime{}$^\dagger$ & \cellcolor{mycolor}\cmark & \cellcolor{mycolor}\Polytime{} (Prop.~\ref{prop:MMS+PO_Chores}) & {\color{opencolor} Open} & {\color{opencolor} Open} & \multirow{-3}{*}{  \cellcolor{mycolor}\cmark{} Thm.~\ref{thm:EFX_Mixed_TopItem}; Cor. \ref{cor:EFX_Mixed_special}}\\
  \hline
  %
  %
  \ldelim\{{4}{20pt}[\RM{} +] & \EF{} & \xmark{} & \Polytime{}$^\dagger$ & \cellcolor{mycolor}\xmark{} & \cellcolor{mycolor}\NPc{} (Thm.~\ref{thm:EF_RM_NP-complete_Chores}) & \xmark{} & \cellcolor{mycolor}\NPc{} (Thm.~\ref{thm:EF_RM_NP-complete_Chores}) & \\
  & \EFX{} & \xmark{} & \NPc{}$^\dagger$ & \cellcolor{mycolor}\xmark & \cellcolor{mycolor}\NPc{} (Thm.~\ref{thm:EFX_RM_NP-complete_Chores}) & \xmark{} & \NPc{}$^\dagger$ & \\
  & \EF{1} & \xmark{} & \NPc{}$^\dagger$ & \cellcolor{mycolor}\xmark{} & \cellcolor{mycolor}\NPc{} (Thm.~\ref{thm:EF1_RM_NP-complete_Chores}) & \xmark{} & \NPc{}$^\dagger$ & \\
  & \MMS{} & \xmark{} & \Polytime{}$^\dagger$ & \cellcolor{mycolor}\xmark{} (Eg.~\ref{eg:MMS_RM_NonExistence}) & \cellcolor{mycolor}\Polytime{} (Thm.~\ref{thm:MMS_RM_Polytime_Chores}) & \xmark & {\color{opencolor} Open} & \\
  \hline
\end{tabular}
\caption{Summary of results for lexicographic preferences. For existence results, a \cmark{} indicates guaranteed existence while a \xmark{} indicates that existence might fail (even for  \textit{objective} instances for mixed items). For computational results, \Polytime{} and \NPc{} refer to polynomial time and \NPC{}, respectively. Results marked by $\dagger$ follow from \citeauthor{HSV+21fair}, and those with $\mathsection$ follow from \citeauthor{ACI+19fair}. Our contributions are highlighted by shaded boxes.}
\label{tab:Results}
\end{table*}

\paragraph{\bf Contributions.}
We undertake a systematic examination of the existential and computational boundaries of fair division under lexicographic preferences. The key conceptual takeaway from our work is that the mixed items problem can be significantly more challenging---both structurally and computationally---than its goods-only counterpart. Below we list some important points of distinction between these models that emerge from our study (also see \Cref{tab:Results}).

\begin{itemize}[leftmargin=*]
    \item \textbf{Envy-freeness}: We show that determining the existence of an envy-free allocation is \NPC{} even for lexicographic chores-only instances (\Cref{thm:EF_Chores_NPC}). By contrast, the goods-only problem is known to admit a polynomial-time algorithm~\cite{HSV+21fair}. Since lexicographic preferences are a subclass of additive valuations, our result extends to the latter domain and strengthens the known hardness results for this problem.
    
    \item \textbf{\EFX{}}: Our main result is that an \EFX{} allocation can fail to exist even for instances with \emph{objective} mixed items (i.e., where each item is either a good for all agents or a chore for all agents)  under lexicographic preferences (\Cref{thm:EFX_counterexample}). This result provides the \emph{first} counterexample for \EFX{} over (any subdomain of) additive valuations (\Cref{cor:EFXNonExistenceAdditive}), and complements the ongoing research effort in understanding the existence of such solutions. By contrast, an \EFX{} allocation always exists for goods~\cite{HSV+21fair}, and we show a similar positive result for the chores-only problem~(\Cref{prop:EFX+PO_Chores} in the appendix).
    
    \item \textbf{\EFX{} and \PO{}}: Given the failure of existence of \EFX{} (and thus \EFX{}+\PO{}) allocations even for objective mixed items, we identify a natural domain restriction where \EFX{}+\PO{} allocations are guaranteed to exist even with \emph{subjective} mixed items and are efficiently computable (\Cref{thm:EFX_Mixed_TopItem}). Notably, our algorithm returns \PO{} solutions despite the failure of the equivalence between \PO{} and sequencibility for mixed items as we observed in \Cref{example:mixed_items}.
    
    \item \textbf{\MMS{}}: When \EFX{} is further weakened to \emph{maximin share} (\MMS{}), we show universal existence and efficient computation for \emph{any} mixed instance~(\Cref{thm:MMS_PO_Polytime_Mixed_Items}).
\end{itemize}

In addition, we consider other notions of fairness and efficiency (such as \EF{1} and rank-maximality) and paint a comprehensive picture of the existential and computational landscape of fair division under lexicographic preferences.

\section{Related Work}

For the goods-only setting with additive valuations, an \EF{1} allocation can be computed through the round-robin algorithm~\cite{CKM+19unreasonable} or the envy-cycle elimination algorithm~\cite{LMM+04approximately}. More importantly, \EF{1} is compatible with \PO{}~\cite{CKM+19unreasonable} and can be computed in pseudo-polynomial time~\cite{BKV18finding}. For the stronger property of \EFX{}, its existence under additive valuations has been an open question. Unfortunately, \EFX{} is known to also be incompatible with \PO{} when valuations are non-negative~\cite{PR20almost}. These negative results no longer hold when the preference domain is restricted to lexicographic preferences. In this domain, not only does an \EFX{} and \PO{} allocation always exist, but one can also be computed efficiently. Furthermore, there is a family of algorithms that can guarantee \EFX{} and \PO{} alongside strategyproofness and other desirable properties~\cite{HSV+21fair}. By contrast, under additive valuations, achieving strategyproofness together with \EF{1} is known to be impossible even for two agents~\cite{amanatidis2017truthful}.

Guaranteeing fairness and efficiency becomes more challenging than its goods-only counterpart when some of the items are chores. For mixed items, \citeauthor{ACI+19fair} showed that with additive valuations, an \EF{1} allocation can still be computed efficiently by the double round-robin algorithm. The (non-)existence of \EFX{} allocations with additive valuations, on the other hand, has been an open question, which we answer in this paper. However, \PO{} seems harder to guarantee together with \EF{1}, and it is not known whether such allocations exist for three or more agents. A notable exception is the chores-only problem with bivalued additive valuations, where an \EF{1} and \PO{} allocation can be computed in polynomial time~\cite{ebadian2022fairly,garg2022fair}.

With additive valuations, an \MMS{} allocation could fail to exist for both the goods-only setting~\cite{KPW18fair} and the chores-only setting~\cite{aziz2017algorithms}. This has given rise to several cardinal \cite{garg2020improved,ghodsi2018fair,aziz2017algorithms} and ordinal \cite{hosseini2021guaranteeing,babaioff2019fair} approximation techniques. For goods-only and chores-only problems with additive valuations, \MMS{} allocations are only known to always exist for restricted domains such as personalized bivalued valuations, and allocations that are \MMS{} and \PO{} can be computed in polynomial time under the restrictions of factored bivalued valuations and weakly lexicographic valuations (allowing for ties between items) \cite{ebadian2022fairly}. However, these results do not imply existence for setting with mixed items, and the constant approximations of \MMS{} may also not always exist~\cite{kulkarni2021indivisible}. These negative results motivate the study of existence and computation of \MMS{} (and its combination with efficiency notions) under restricted domains such as lexicographic preferences. 

The term ``mixed'' has also been used to refer to mixture of indivisible and divisible resources in the literature~\cite{BLL+21fair,BSV21approximate}, but in this paper we only consider mixture of indivisible items (goods and chores).

\section{Preliminaries} \label{sec:prel}
\paragraph{Model.}  For any $k\in\mathbb{N}$, we define $[k]\coloneqq\{1,\dots,k\}$. An instance of the allocation problem with {\em mixed items} is a tuple $\langle N,M,G,C,\id \rangle$ where $N\coloneqq[n]$ is a set of $n$ {\em agents} and $M$ is a set of $m$ {\em items} $\{o_1,\dots,o_m\}$. Here, $G\coloneqq(G_1,\dots,G_n)$ and $C\coloneqq(C_1,\dots,C_n)$ are collections of subsets of $M$, where, for each $i\in [n]$, $G_i\subseteq M$ is the set of \emph{goods} and $C_i=M\setminus G_i$ is the set of {\em chores} for agent $i$, respectively. Additionally, $\id \, \coloneqq (\id_1,\dots,\id_n)$ is a {\em preference profile} that specifies for each agent $i\in N$ an \emph{importance ordering} $\id_i\in\ml$ over the individual items in $M$ in the form of a linear order; here $\ml$ is the set of all (strict and complete) linear orders over $M$ (all goods and chores). For example, we write $o^+_1 \, \id_i \, o^-_2 \, \id_i \, o^+_3$ to indicate that agent $i$ considers items $o_1$ and $o_3$ as goods and the item $o_2$ a chore, and ranks $o_1$ above $o_2$ and $o_2$ above $o_3$.\footnote{Not to be interpreted as ``agent $i$ \emph{prefers} chore $o_2$ over good $o_3$''; see the paragraph on `Lexicographic Preferences' for that.}

We use $\id_i(k)$ to denote the $k$-th ranked item in the importance ordering of agent $i$, and $\id_i(k, S)$ to specify the $k$-th ranked item for agent $i$ among items in set $S$. Thus, in the above example, $\id_i(1)=o_1$ and $\id_1(1,\{o_2,o_3\})=o_2$. 

In an instance with \textbf{objective} mixed items, each item is either a good for all agents or a chore for all agents, i.e., for any pair of agents $i,j\in N$, $G_i=G_j$ and $C_i=C_j$. In a {\em goods-only} (respectively, {\em chores-only}) instance, every item is a good (respectively, a chore) for all agents, i.e., for every agent $i\in N$, $G_i=M$ (respectively, $C_i=M$).

\paragraph{Bundles.} A {\em bundle} is any subset $X\subseteq M$ of the items. Given any bundle $X\subseteq M$, we will write $X^{i+}\coloneqq X\cap G_i$ and $X^{i-}\coloneqq X\cap C_i$ to denote the sets of goods and chores in $X$, respectively, according to agent $i$. 

\paragraph{Allocations.} An {\em allocation} $A=(A_1,\dots,A_n)$ is an $n$-partition of $M$, where $A_i\subseteq M$ is the bundle assigned to agent $i$. We will write $\Pi(M)$ to denote the set of all $n$-partitions of $M$. We say that an allocation $A$ is {\em partial} if $\bigcup_{i\in N}A_i\subset M$, and {\em complete} if $\bigcup_{i\in N}A_i=M$. Unless stated explicitly otherwise, an `allocation' will refer to a complete allocation. 

\paragraph{Lexicographic Preferences.} We will assume that agents' preferences 
over bundles are given by the lexicographic extension of their importance orderings $\id\coloneqq(\id_1,\dots,\id_n)$, which are over the individual items, taking into account whether an item is considered a good or a chore. Informally, this means that an agent with importance ordering $o^+_1\id o^-_2\id o^+_3$ prefers any bundle that contains the good $o_1$ over any bundle that does not, subject to that, it prefers a bundle that \emph{does not} contain the chore $o_2$ over any other bundle that contains $o_2$, and so on. The importance ordering $o^+_1\id o^-_2\id o^+_3$ over individual items induces the ranking $\>_i$ over the bundles given by $\{o_1^+,o_3^+\} \, \> \, \{o_1^+\} \, \> \, \{o_1^+,o_2^-,o_3^+\} \, \> \, \{o_1^+,o_2^-\} \, \> \, \{o_3^+\} \, \> \, \emptyset \, \> \, \allowbreak\{o_2^-,o_3^+\} \, \> \, \{o_2^-\}$, where $\emptyset$ denotes the empty bundle. 

Formally, given any pair of bundles $X,Y\subseteq M$, we say that agent $i\in N$ prefers bundle $X$ to $Y$, denoted as $X\ \>_i \ Y$, if and only if either (i) there exists a good $g\in G_i\cap (X\setminus Y)$ such that $\{o^+\in Y\cap G_i:o^+\id_i g\}\subset X$, or (ii) there exists a chore $c\in C_i\cap (Y\setminus X)$ such that $\{o^-\in X\cap C_i:o^-\id_i c\}\subset Y$. For any agent $i\in N$, and any pair of bundles $X,Y\subseteq M$, we will write $X\succeq_i Y$ if either $X\>_i Y$ or $X=Y$.

\label{subsec:Prelimns_Fairness_Notions}

\paragraph{Envy-Freeness.} 
An allocation $A$ is (a) \emph{envy-free} (\EF{}) if for every pair of agents $i,h\in N$, $A_i\succeq_i A_h$, (b) \emph{envy-free up to one item} (\EF{1}) if for every pair of agents $i,h\in N$ such that $A_i^{i-}\cup A_h^{i+} \neq \emptyset$, there exists an item $o\in A_i^{i-}\cup A_h^{i+}$ such that either $A_i\succeq_i A_h\setminus\{o\}$ or $A_i\setminus\{o\}\succeq_i A_h$, and (c) \emph{envy-free up to any item} (\EFX{}) if for every pair of agents $i,h\in N$ such that $A_i^{i-}\cup A_h^{i+} \neq \emptyset$ and for {\em every} item $o\in A_h^{i+} \cup A_i^{i-}$, it must be that (i) if $o\in A_h^{i+}$, $A_i\succeq_i A_h\setminus\{o\}$ and (ii) if $o\in A_i^{i-}$, $A_i\setminus\{o\}\succeq_i A_h$. In \Cref{app:EFX_variation}, we define two relaxations of \EFX{}, namely \EFX{}-c and \EFX{}-g. where only chores (resp., only goods) can be removed. Interestingly, our counterexample for \EFX{} (\Cref{thm:EFX_counterexample}) holds even for \EFX{}-c.

\paragraph{Maximin Share.} An agent's maximin share is its most preferred bundle that it can guarantee itself as a divider in an $n$-person cut-and-choose procedure against adversarial opponents~\cite{B11combinatorial}. Formally, the maximin share of agent $i$ is given by $\MMS_i \coloneqq \max_{P \in \Pi(M)} \min_{i} \{P_1,\allowbreak\dots,P_n\}$, where $\min\{\cdot\}$ and $\max\{\cdot\}$ denote the least-preferred and most-preferred bundles with respect to $\>_i$. An allocation $A$ satisfies \emph{maximin share} (\MMS{}) if each agent receives a bundle that it weakly prefers to its maximin share, i.e., for every $i \in N$, $A_i \succeq_i \MMS_i$.

\paragraph{Pareto Optimality.}
Given a preference profile $\id$, an allocation $A$ is said to be \emph{Pareto optimal} (\PO{}) if there is no other allocation $B$ such that $B_i \succeq_i A_i$ for every agent $i \in N$ and $B_h \succ_h A_h$ for some agent $h \in N$. To avoid vacuous solutions such as leaving all chores unassigned, we will always require a Pareto optimal allocation to be complete.

A \emph{picking sequence} of length $k$ is an ordered tuple $\tau=\langle s_1,s_2,\dots,s_k\rangle$ where, for each $i \in [k]$, $s_i\in N$ denotes the agent who picks its favorite available item, that is, its top-ranked remaining good, if one exists, or otherwise its bottom-ranked remaining chore as per its importance ordering $\vartriangleright_i$. A \emph{sequencible} allocation is one that can be simulated via a picking sequence. It is known that for goods-only instances, an allocation is sequencible if and only if it is \PO{}~\cite{beynier2019efficiency}. We observe that the equivalence also holds for \emph{chores-only} problems.

\begin{restatable}[\textbf{\PO{} $\Leftrightarrow$ sequencible for chores}]{prop}{POSequencibleChores}
An allocation of chores is \PO{} if and only if it is sequencible.
\label{prop:po_seqencible_chores}
\end{restatable}

However, when dealing with \emph{mixed} items, sequencibility is no longer a sufficient condition for guaranteeing \PO{}.

\begin{restatable}[\textbf{\PO{} and sequencibility for mixed items}]{prop}{POSequencibleMixed}
For mixed items, Pareto optimality implies sequencibility, but the converse is not true even for objective mixed items.
\label{prop:po_seqencible_mixed}
\end{restatable}

To see why sequencibility does not imply \PO{}, consider the objective mixed items instance with three items $\{o^+_1,o^+_2,o^-_3\}$ and two agents where agent $1$'s importance ordering is $o^+_1\id\,o^+_2\id\,o^-_3$, and agent $2$'s ordering is $o^-_3\id\,o^+_1\id\,o^+_2$. The picking sequence $\langle 1,2,2 \rangle$ allocates $\{o^+_1\}$ to agent $1$ and $\{o^-_3,o^+_2\}$ to agent $2$. However, this allocation is Pareto dominated by the allocation that gives all items to agent $1$.

Given an instance with mixed items, there always exists a Pareto optimal allocation (since there are only finitely many allocations). Further, one such allocation can be computed in polynomial time. This can be shown by observing that a \emph{rank-maximal} allocation is Pareto optimal; see \Cref{sec:Preliminaries_Appendix}.

All missing proofs and detailed algorithms appear in the appendic.

\section{Results}
\label{sec:Results}

We start our investigation by considering the strongest fairness notion---\emph{envy-freeness}. As we will see, this notion will provide us the first contrast between goods and chores.

\subsection{Envy-Freeness}
\label{subsec:EF}

With indivisible items, a complete and envy-free allocation may not always exist. Thus, it is of interest to ask whether one can efficiently determine the existence of such solutions. This problem admits a polynomial-time algorithm in case of goods~\cite{HSV+21fair}, but turns out to be \NPC{} for chores, and by extension, for mixed items~(\Cref{thm:EF_Chores_NPC}). 

\begin{restatable}[\textbf{\EF{} for chores}]{thm}{EFChoresNPHard}
Determining whether a chores-only instance with lexicographic preferences admits an envy-free allocation is \NPC{}.
\label{thm:EF_Chores_NPC}
\end{restatable}

To understand the sharp contrast in the complexity of the goods and chores problems, recall that for goods, an allocation is envy-free if and only if each agent gets its top-ranked item.
One can efficiently check whether there exists a partial allocation satisfying this property via a straightforward matching computation. Furthermore, if such a partial allocation exists, \emph{any} completion of it is also envy-free. 

By contrast, envy-freeness for chores entails that for every agent, the \emph{worst} or least-preferred chore (i.e., highest-ranked in the importance ordering) in its own bundle is strictly preferred over the worst chore in any other agent's bundle. Thus, given an envy-free partial allocation, its completion may no longer be envy-free since, upon receiving more items, a different chore could become the worst.

We note that the allocation constructed in the forward direction in the proof of \Cref{thm:EF_Chores_NPC} is sequencible and therefore also Pareto optimal (\Cref{prop:po_seqencible_chores}), implying that the hardness result also holds for \EF{}+\PO{} allocations (\Cref{sec:EF+PO_Chores}).

\subsection{Envy-Freeness up to any Item (\EFX{})}
\label{subsec:EFX}

Let us now turn our attention to a relaxation of envy-freeness called envy-freeness up to any item (\EFX{}). Prior work has shown that an \EFX{} and Pareto optimal allocation always exists for goods~\cite{HSV+21fair}. In \Cref{prop:EFX+PO_Chores} in the appendix, we show that a similar positive result can be achieved for the chores-only problem via the following simple procedure: Fix a priority ordering $\sigma$ over agents. Let the first agent in $\sigma$ pick its most preferred $m -n$ chores. Then, all agents (including the first agent) pick one chore each according to $\sigma$ from the remaining items.

Our main result in this section is that the above positive results for goods-only and chores-only models fail to extend to the mixed items setting: We show that an \EFX{} allocation may not exist even for \emph{objective} mixed items, i.e., when each item is either a common good or a common chore~(\Cref{thm:EFX_counterexample}).

\begin{restatable}[\textbf{Non-existence of \EFX{}}]{thm}{EFXcounterexample}
There exists an instance with objective mixed items and lexicographic preferences that does not admit any \EFX{} allocation.
\label{thm:EFX_counterexample}
\end{restatable}

Since lexicographic preferences are a subclass of additive valuations, our counterexample also shows that an \EFX{} allocation fails to exist under non-monotone and additive valuations~(\Cref{cor:EFXNonExistenceAdditive}).\footnote{A valuation function $v_i:2^M\rightarrow\mathbb{R}$ is \emph{non-monotone} if for some subsets $T \subset S \subseteq M$, we have $v_i(T) > v_i(S)$ and for some (possibly different) subsets $T' \subset S' \subseteq M$, we have $v_i(T') < v_i(S')$.} Our result complements that of \citeauthor{BBB+20envy} who showed that an \EFX{} allocation could fail to exist for two agents with non-monotone, {\em non-additive}, and identical utility functions.

\begin{restatable}[]
{cor}{EFXNonExistenceAdditive}
An \EFX{} allocation can fail to exist for instances with non-monotone and additive valuations.
\label{cor:EFXNonExistenceAdditive}
\end{restatable}

The counterexample in the proof of \Cref{thm:EFX_counterexample} (given below) uses only \emph{four} agents and \emph{seven} items. Interestingly, for the said number of agents and items, an \EFX{} allocation is guaranteed to exist for \emph{goods-only} instances even under \emph{monotone} valuations~\cite{M21extension}, which is significantly more general than additive (or lexicographic) preferences.

It is also known that when agents belong to one of two given ``types'', an \EFX{} allocation is guaranteed to exist for \emph{goods-only} instances under \emph{monotone} valuations~\cite{M21extension}. Our result in \Cref{thm:EFX_counterexample}, which also has two types of agents, demonstrates a barrier to extending this result in the non-monotone setting, even under lexicographic preferences.

\begin{proof} (of \Cref{thm:EFX_counterexample})
Consider an objective mixed items instance with four agents. Agents $1$ and $2$ have the same importance ordering, and so do agents $3$ and $4$, as shown below:

\newcommand{\gooditem}[1]{o_{#1}^+}
\newcommand{\choreitem}[1]{o_{#1}^-}

\begin{center}
\begin{tabular}{l l@{$\rhd$\ } l@{$\rhd$\ } l @{$\rhd$\ } l@{$\rhd$\ } l@{$\rhd$\ } l@{$\rhd$\ } l}
$1, 2 \, :$ & $\choreitem{2}$ & $\choreitem{3}$ & $\choreitem{4}$  & $\gooditem{1}$ & $\choreitem{5}$ & $\choreitem{6}$ & $\choreitem{7}$ \\
$3, 4 \, :$ & $\choreitem{5}$ & $\choreitem{6}$ & $\choreitem{7}$  & $\gooditem{1}$ &  $\choreitem{2}$ & $\choreitem{3}$ & $\choreitem{4}$\\
\end{tabular}
\end{center}

Since the items are objective, we will find it convenient to use the phrases `the good $o_1$' and `the chore $o_2$' instead of just calling them `items'.

Suppose, for contradiction, that an \EFX{} allocation exists. Without loss of generality, suppose agent $1$ gets the good $o_1^{+}$, and
 
let $A_i$ denote the bundle allocated to agent $i$. We will show a contradiction via case analysis, depending on the chores allocated to agent $1$.\\

\textbf{Case 1:} If $A_1 \cap \{o_2^{-},o_3^{-},o_4^{-}\} = \emptyset$. That is, agent $1$'s  allocated chores are a (possibly empty) subset of $\{o_5^{-},o_6^{-},o_7^{-}\}$, which are all ranked below $\gooditem{1}$ according to agent $1$'s importance ordering.\\

This means that regardless of what agent $2$ gets, it prefers the bundle $A_1$ to its own bundle $A_2$. Therefore, $A_2$ must be empty, as otherwise agent $2$ will prefer $A_1$ even when some chore is removed from $A_2$. Thus, the chores $o_2^{-},o_3^{-}$, and $o_4^{-}$ must be allocated to agents $3$ and $4$, which means that one of these agents must get at least two of these chores. Suppose, without loss of generality, that agent $3$ gets at least two chores. Then, agent $3$ would prefer the empty bundle $A_2$ after any chore is removed from $A_3$, a contradiction to \EFX{}.\\

\textbf{Case 2:} If $ A_1 \cap \{o_5^{-},o_6^{-},o_7^{-}\} = \emptyset$. That is, agent $1$'s allocated chores are a subset of $\{o_2^{-},o_3^{-},o_4^{-}\}$, which are all ranked above $\gooditem{1}$ according to agent $1$'s importance ordering.\\

This means that regardless of how the remaining chores are assigned, both agents $3$ and $4$ will strictly prefer $A_1$ over their respective bundles (because their most important item in $A_1$ is $\gooditem{1}$). Now, if agent $3$ or $4$ is assigned any item, which must be a chore, then even after removing this chore, it would still envy $A_1$. Therefore, agents $3$ and $4$ cannot be allocated any item. This means that agent $2$ gets at least $\{\choreitem{5},\choreitem{6},\choreitem{7}\}$, which implies that after any item (which must be a chore) is removed from agent $2$'s bundle, agent $2$ envies agent $3$ (who is not allocated any item). This contradicts \EFX{}.\\

\textbf{Case 3:} If $A_1 \cap \{o_2^{-},o_3^{-},o_4^{-}\} \ne \emptyset$ and $A_1 \cap \{o_5^{-},o_6^{-},o_7^{-}\} \ne \emptyset$. That is, agent $1$ gets at least one chore above good $\gooditem{1}$ and at least one chore below $\gooditem{1}$ according to its importance ordering.\\
    
Choose any $x\in A_1 \cap \{o_2^{-},o_3^{-},o_4^{-}\}$ and $y\in A_1 \cap \{o_5^{-},o_6^{-},o_7^{-}\}$. Then, because of \EFX{}, agent $1$ should not prefer any other agent's bundle after $y$ is removed from $A_1$. This means that for any  $i\in \{2,3,4\}$, $A_i$ must contain a chore that is ranked higher than $x$ according to agent $1$'s importance order. However, there are at most two chores perceived to be ranked higher than $x$ by agent $1$, which contradicts \EFX{}.
\end{proof}

\subsection{\EFX{} and Pareto optimality}
\label{subsec:EFX+PO}

We have seen that an \EFX{} allocation may not exist for mixed items. This negative result prompts us to identify a subclass of lexicographic instances with subjective mixed items for which an \EFX{} and Pareto optimal allocation is guaranteed to exist. Specifically, we will now require that there be an agent whose \emph{top-ranked} item in its importance ordering is a good~(\Cref{thm:EFX_Mixed_TopItem}).

\begin{algorithm}[t]
\DontPrintSemicolon
 \linespread{1}
\KwIn{A lexicographic mixed instance $\langle N,M,G,C,\id \rangle$}
\KwOut{An \EFX{}+\PO{} allocation $A$}
Select an arbitrary agent $i\in N$ such that $\id_i(1) \in G_i$ \;\label{line:EFX+PO_mixed_first_agent}
Let $C' \coloneqq \{o\in M \, : \, \forall j\in N\setminus \{i\}, o \in C_j \}$ \tcp*{{The set of all common chores for the remaining agents.}}
$A_i \leftarrow$ $\id_i(1)\ \cup\ C'$ \; 
$N \leftarrow N \setminus \{i\}$\;
$M \leftarrow M \setminus A_i$\;\label{line:EFX+PO_mixed_end_first_agent}
\Comment{{The remaining instance has no common chore.}}
\While{there exists an unallocated item}{ \label{line:no_common_chores}
\If{$|N| = 1$}{
Assign all items to the remaining agent \;
}
\Else{
Find the smallest $k \in \{1,2,\dots,|M|\}$ such that the set $S^k \coloneqq \{ i \in N \,: \, \id_i(k) \in G_i$ \} is non-empty \tcp*{{
set of agents whose $k^\text{th}$-ranked item is a good.}}\label{line:EFX+PO_mixed_smallest_k}

Select any agent $j \in S^{k}$ \;\label{line:EFX+PO_mixed_select_agent}
$C' \coloneqq \{o\in M\, : \, \forall i\in N \setminus \{j\}, o \in C_i\}$ \;\label{line:EFX+PO_mixed_find_chores}
$A_j \leftarrow \{\id_{j}(k)\} \cup C'$ \;\label{line:EFX+PO_mixed_allocate_chores}
$N \leftarrow N \setminus \{j\}$\;\label{line:EFX+PO_mixed_eliminate_agent}
$M \leftarrow M \setminus A_j$\;
}
}
\KwRet{$A$}
\caption {Finding an \EFX{}+\PO{} allocation when there is an agent whose top-ranked item is a good.
}
\label{alg:EFX+PO_mixed}
\end{algorithm}

\begin{restatable}[\textbf{\EFX{}+\PO{} when some agent has a  top-ranked good}]{thm}{EFXMixedTopItem}
Given a lexicographic mixed instance where some agent's top-ranked item is a good, an \EFX{}+\PO{} allocation always exists and can be computed in polynomial time.
\label{thm:EFX_Mixed_TopItem}
\end{restatable}

\begin{proof} (sketch) Let us start by discussing why the allocation returned by our algorithm is \EFX{}, followed by a similar discussion for \PO{}.

\paragraph{Guaranteeing \EFX{}.}Intuitively, the assumption about some agent's top-ranked item being a good allows us to deal with the common chores without violating \EFX{} as follows (see Algorithm~\ref{alg:EFX+PO_mixed}): An agent whose top-ranked item is a good can be given that item together with all items that are common chores for the rest of the $n-1$ agents. Since the preferences are lexicographic, this agent will not envy any other agent regardless of how the remaining items are allocated. 

The first agent is now eliminated from the instance along with its assigned bundle. Observe that the reduced instance (with $n-1$ agents) has \emph{no common chore}, that is, each item is considered as a good by at least one agent. The algorithm now uses the following strategy iteratively: It identifies an agent with the highest-ranking good (say agent $j$ and good $g$), gives good $g$ to agent $j$ together with the common chores of the remaining $n-2$ agents, and then eliminates agent $j$.

Note that since agent $j$ receives its highest-ranked good among the remaining items, it will not envy any agent that is eliminated \emph{after} it, regardless of how the remaining items are assigned. Furthermore, by the `no common chores' property, any item that is a chore for the rest of the agents must be a good for agent $j$. This means that agent $j$ only receives those items that it considers to be goods. Thus, when evaluating \EFX{} from agent $j$'s perspective, we only need to look at the items in other agents' bundles that agent $j$ considers to be goods. For any agent that was eliminated \emph{before} $j$, there can be at most one such item (by virtue of assigning common chores), and thus \EFX{} is maintained.

\paragraph{Guaranteeing \PO{}.}
Suppose, for contradiction, that the allocation $A$ returned by Algorithm~\ref{alg:EFX+PO_mixed} is Pareto dominated by the allocation $B$. We will argue by induction that for every agent $i$, we must have $A_i \subseteq B_i$, which would contradict Pareto optimality since $A$ and $B$ must be distinct. For ease of discussion, let us name the agents according to the order in which they are eliminated by Algorithm~\ref{alg:EFX+PO_mixed}.

Recall from the above discussion on \EFX{} that the most important item in each agent's bundle under $A$ must be a good for it, i.e., $\id_i(1,A_i)\in G_i$. We will argue that each agent $i$ retains this item under $B$, i.e., $\id_i(1,A_i)\in B_i$. It is easy to see that agent $1$ must retain $\id_1(1,A_1)$ since it is also $1$'s most important item overall. Suppose the induction hypothesis is true for agents $1,\dots,i-1$. Then, it must also hold for agent $i$ because (i) for every agent $h$ eliminated {\em before} $i$, $A_h\setminus\{\id_h(1,A_h)\}\subseteq C_i$, (ii) for every item $o\in A_j$ allocated to an agent $j$ eliminated {\em after} $i$, $\id_i(1,A_i)\id_i o$, and (iii)  $A_i\setminus\id_i(1,A_i)\subseteq G_i$ as we argued above. Together this means that if $\id_i(1,A_i)\not\in B_i$, then $A_i \, \>_i \, B_i$, a contradiction.

A similar inductive argument shows that every agent $i$ must retain all other items from $A_i$ in $B_i$. 
Indeed, the last agent, say $k$, receives all remaining items (i.e., $A_k\subseteq G_k$) and all items allocated to an agent $h$ eliminated {\em before} $k$ are chores for $k$. By a similar argument, if the induction hypothesis holds true for agents $k,k-1,\dots,i-1$, it must also hold for an agent $i>1$. Now, for agent $i=1$, since for every agent $j\in N\setminus\{1\}$, $B_j=A_j$ as we argued above, and every item in $A_1\setminus\{\id_1(1,A_1)\}$ is a chore for $j$, we must have that if for every $i\in N$, $B_i\succeq_i A_i$, then $B_i=A_i$. This contradicts our assumption that $B$ Pareto dominates $A$.
\end{proof}

Another special case where an \EFX{}+\PO{} allocation is guaranteed to exist is when every item is considered a good by at least one agent, i.e., there are no \emph{common chores}.

\begin{restatable}[\textbf{\EFX{}+\PO{} for mixed instances without common chores}]{cor}{EFXMixedspecial}
Given a lexicographic mixed instance without any common chore, an \EFX{}+\PO{} allocation always exists and can be computed in polynomial time.
\label{cor:EFX_Mixed_special}
\end{restatable}

\subsection{Maximin Share (\MMS{})}
\label{subsec:MMS}

In light of the failure in guaranteeing \EFX{} even for objective mixed items, we investigate the existence of \MMS{} allocations for mixed items. We show that not only does an \MMS{} allocation exist for \textit{subjective} mixed items, but also that such an allocation can be computed efficiently.

We start by characterizing \MMS{} bundles by examining the structure of an agent's maximin share. Given a lexicographic mixed instance, an agent's maximin share is identified by its top-ranked item: if agent $i$'s top-ranked item is a good, $\MMS{}_i$ is either an empty set (when the number of goods is less than that of agents), or it is the set of the least preferred $n-1$ goods. Otherwise, when agent $i$'s top-ranked item is a chore, then $\MMS{}_i$ is uniquely defined by the union of the top-ranked item (worst chore) and all the goods.

\begin{restatable}[\textbf{Characterizing \MMS{} for mixed items}]{prop}{MMSmixed}
Given an instance $\langle N,M,G,C,\id\rangle$ with lexicographic mixed items, the maximin share of agent $i$ can be defined based on whether its top-ranked item is a good or a chore, as follows:

$$\text{MMS}_i = \begin{cases}
        G_i \setminus \{\id_{i}([n-1], G_i)\}, & \text{if}\ \id_i(1) \in G_i \wedge |G_{i}| \geq n \\
        \emptyset, & \text{if}\ \id_i(1) \in G_i \wedge |G_{i}| < n \\ 
        \id_{i}(1, C_i) \cup G_i, & \text{if}\ \id_i(1) \in C_i.
        \end{cases}
$$
\label{prop:MMS_mixed_char}
\end{restatable}

Although \EFX{} may not always exist for mixed items (\Cref{thm:EFX_counterexample}), we show that whenever such an allocation exists, it also satisfies \MMS{}.
Note that the converse does not hold, that is, even for chores-only instances (where \EFX{} always exists), \MMS{} does not imply \EFX{} (\cref{prop:EFX_implies_MMS} in the appendix).

\begin{restatable}[\textbf{\EFX{}$\implies$\MMS{} for mixed items}]{prop}{EFXImpliesMMSForMixed}
For mixed items under lexicographic preferences, an \EFX{} allocation (whenever it exists) satisfies \MMS{}, but the converse is not always true.
\label{prop:EFX_MMS_mixed}
\end{restatable}

We develop an algorithm that computes an \MMS{} allocation for \textit{any} lexicographic instance---even with subjective mixed items---in polynomial time.

\paragraph{Description of algorithm.}
Our algorithm (Algorithm~\ref{alg:MMS_mixed} in the appendix) first identifies the set $C'$ of all common chores and proceeds in two steps: In {\bf Step 1}, all common chores are allocated without violating \MMS{}, and in {\bf Step~2}, all remaining items are allocated as goods.\\

\noindent{\bf Step 1.}  {\em If there exists an agent whose top-ranked item is a good,} i.e., then run Algorithm \ref{alg:EFX+PO_mixed} to achieve an \EFX{} allocation. By \cref{prop:EFX_MMS_mixed}, \EFX{} implies \MMS{} for mixed items.

{\em Otherwise, if every agent's top item is a chore,} a priority ordering $\sigma$ over agents is fixed, and a serial dictatorship is run where agent $\sigma_1$ picks its most preferred $|C'| - n + 1$ chores from the set of all common chores, $C'$, and the remaining agents each pick one remaining chore from $C'$. Note that if $|C'| < n$, the first $n - |C'|$ agents pick one chore and the rest receive nothing. If an agent $k$ receives its worst chore from $C_k$, it is given its remaining goods in $G_k$.

\noindent{\bf Step 2.} All remaining items are allocated through a picking sequence. In each turn, an agent picks all remaining items it considers as goods, or picks nothing. All remaining items are only allocated as goods, and thus, do not violate \MMS{}.

\begin{restatable}[\textbf{\MMS{} for mixed items}]{thm}{MMSPOMixedItems}
Given a lexicographic mixed instance, there is a polynomial-time algorithm that computes an \MMS{} allocation even for subjective items.
\label{thm:MMS_PO_Polytime_Mixed_Items}
\end{restatable}

The significance of \cref{thm:MMS_PO_Polytime_Mixed_Items} stems from providing an efficient algorithm for computing an \MMS{} allocation for any lexicographic mixed instance (including subjective instances). Yet, the problem of computing an \MMS{}+\PO{} allocation remains open even for objective instances.

\section{Concluding Remarks}

We studied the interaction between fairness and efficiency for mixture of indivisible goods and chores under lexicographic preferences. We showed that an \EFX{} allocation may not always exist for mixed items. Nonetheless, we identified natural classes of lexicographic instances for which an \EFX{}+\PO{} allocation exists and can always be computed efficiently. We further proved an \MMS{} allocation always exists and can be computed efficiently even for subjective mixed instances.

Going forward, it will be interesting to resolve the computational complexity of checking the existence of \EFX{} allocations for mixed items. Another relevant direction will be to explore the space of strategyproof mechanisms satisfying desirable fairness and efficiency guarantees.

\section*{Acknowledgements}

Hadi Hosseini acknowledges support from NSF IIS grants \#2052488 and \#2107173. Lirong Xia acknowledges support from NSF IIS grants \#1453542, \#171633, and \#2107173.

We are grateful to anonymous reviewers for their valuable feedback.

\small
\bibliographystyle{named}
\bibliography{mixed}

\clearpage
\newpage
\normalsize
\appendix
\section*{Appendix}
\label{sec:Appendix}

\section{Omitted Material from \cref{sec:prel}}
\label{sec:Preliminaries_Appendix}

\POSequencibleChores*

\begin{proof}
For chores, we know that \PO{} implies sequencibility due to \Cref{prop:po_seqencible_mixed}. 

Here, we prove that sequencibility implies \PO{} for chores. Suppose for the sake of contradiction that an allocation $A$ is the outcome of a picking sequence, but is not \PO{}. Then there exists an allocation $B$ which Pareto dominates $A$. Since agents have lexicographic preferences, this means that there exist agents $N'\subset N$ who receive strictly better bundles in $B$ than in $A$, and the allocations of agents in $N\setminus N'$ are identical in $A$ and $B$. Let $N^*=\{i\in N':B_i\setminus A_i\neq\emptyset\}$. For every agent $i\in N^*$, let $K_i=B_i\setminus A_i$ and $L_i=A_i\setminus B_i$, and let $o^-_i$ be agent $i$'s highest-ranked (worst) chore in $L_i$ according to its importance order $\id_i$.

Since $B$ Pareto dominates $A$ and agents have lexicographic preferences, it follows that every agent $i\in N^*$, ranks $o^-_i$ higher than every chore in $K_i$. It is easy to see that $o_i^-$ must be allocated to an agent $j\in N^*$, since $B$ is a complete allocation, and by our assumption that $B$ Pareto dominates $A$ and lexicographic preferences. 

Let $H$ be the graph where there is a node for each agent $i\in N^*$, and there is an edge $(i,j)$ if and only if $o^-_i\in B_j$. Since every node in $H$ has an outgoing edge, $H$ must have a cycle. Let $C$ be one such cycle in $H$ involving $k$ agents, and w.l.o.g. let the agents be named so that for each $i\in [k]$, there is an edge $(i,i+1\pmod{k})$. 

For every agent $i\in [k]$, let $r_i$ be the round in the picking sequence that produces allocation $A$ in which they pick $o^-_i$. For each $i\in [k]$, since agent ${i+1\pmod{k}}$ prefers $B$ to $A$, it follows that agent $i+1\pmod{k}$ ranks $o^-_i$ lower than its highest ranked chore in $A$. This implies that at round $r_{i+1\pmod{k}}$, $o^-_i$ was unavailable, which implies that $r_i<r_{i+1\pmod{k}}$. However, this implies that $r_1<r_2<\dots<r_k<r_1$, which is a contradiction.
\end{proof}

\POSequencibleMixed*
\begin{proof}
{\bf \PO{} implies sequencibility.} 
Suppose for the sake of contradiction, there exists a \PO{} allocation $A$ which is not sequencible. Suppose $k\in\{0,1,\dots,m-2\}$ is the largest number such that there exists a picking sequence of length $k$ which induces a partial allocation $B$ of $k$ items that agrees with $A$. Let $M'$ be the items allocated by $B$. Now, by our choice of $k$, it must hold that in the allocation $B$, no agent receives either their top-ranked good in $M\setminus M'$ if one exists, or their bottom-ranked chore in $M\setminus M'$ otherwise. For each agent $i\in N$, let $o_i^*$ be either its top-ranked good in $G_i\cap(M\setminus M')$ if it is non-empty, or the bottom-ranked chore in $C_i\cap(M\setminus M')$ otherwise.

Now, let $H$ be the graph where there is a node corresponding to each agent $i$, and there is an edge from agent $i$ to agent $j$ if $o_i^*\in A_j$. Notice that every agent must have an outgoing edge in $H$, and therefore $H$ must have at least one cycle, and that the cycles in $H$ are disjoint. 

Let $C$ be the partial allocation of items in $M\setminus M'$ constructed from one such cycle in $H$ as follows: For every pair of agents $i,j$ such that the edge $(i,j)$ is involved in the cycle, agent $i$ is allocated $o_i^*$. All remaining items in $M\setminus M'$, are allocated to the same agents as they are in allocation $A$.

It is easy to see that by our assumption of lexicographic preferences, the complete allocation constructed by allocating items in $M'$ and $M\setminus M'$ according to the partial allocations $B$ and $C$ respectively Pareto dominates $A$. This is a contradiction to our assumption that $A$ is Pareto optimal.

\paragraph{Sequencibility does not imply \PO{} even for objective mixed items.} Consider the objective mixed items instance with two agents:

\begin{align*}
1: ~&~ o_1^+ \ \id \ o_2^+ \ \id \ o^-_3\\ \nonumber
2: ~&~ o^-_3 \ \id \ o^+_1 \ \id \ o^+_2\\
\end{align*}

The picking sequence $1,2,2$ allocates $\{o^+_1\}$ to agent $1$ and $\{o^+_2,o^-_3\}$ to agent $2$. It is easy to see that this allocation is Pareto dominated by the allocation where agent $1$ is given all items.
\end{proof}

\subsection{Rank-maximality for Mixed Items}

We start by providing a formal definition for rank-maximality when allocating mixed items.

\noindent{\bf Rank-Maximality.} A \emph{rank-maximal} (\RM{}) allocation~\cite{irving2006rank,paluch2013capacitated} is one that maximizes the number of agents who receive their highest-ranked good, subject to which it maximizes the number of agents who receive their second-highest good, and so on, subject to which it maximizes the number of agents who receive their lowest-ranked chore, subject to which it maximizes the number of agents who receive the second-lowest chore, and so on. Given an allocation $A$, its \emph{signature} refers to a tuple $(n^+_1,n^+_2,\dots,n^+_m,n^-_1,n^-_2,\dots,n^-_m)$ where $n^+_k=|\{i\in N:\id_i(k,G_i)\in A_i\}|$ is the number of agents who receive their $k$-th highest ranked good and $n^-_k=|\{i\in N:\id_i(m-k,C_i)\in A_i\}|$ is the number of agents who receive their $k$-th lowest ranked chore (equivalently, $(m-k)$-th highest ranked chore). 

It is easy to see that computing \emph{some} rank-maximal allocation for a given instance, as well as verifying whether a given allocation is rank-maximal, can be done in polynomial time. A rank-maximal allocation for mixed items implies Pareto optimality by definition, which gives a polynomial-time algorithm for computing a \PO{} allocation.

\begin{restatable}[\textbf{\PO{} allocation of mixed items}]{prop}{POMixed}
Given an instance with mixed items under lexicographic preferences, a Pareto optimal allocation can be computed in polynomial time.
\label{prop:po_mixed}
\end{restatable}

\section{Omitted Material from Section~\ref{subsec:EF} on \texorpdfstring{\EF{}}{EF}} \label{app:EF}

\subsection{Proof of Theorem~\ref{thm:EF_Chores_NPC}}

\EFChoresNPHard*
\begin{proof}
Membership in \NP{} is easy to verify. We will show \NPC{}ness by a reduction from {\sc SAT}, where we are given a CNF $F=C_1\wedge \cdots \wedge C_s$ of $s$ clauses over $t$ binary variables $Y_1,\ldots,Y_t$, and we are asked whether a satisfying assignment exists. Given any {\sc SAT} instance, we construct a chores instance as follows.

{\bf Agents:} There are $4t$ agents, denoted by 
$$\{X_i,\neg X_i,X_i^*,\neg X_i^*:i\in [t]\}$$

{\bf Chores:} There are $s+5t$ chores\footnote{Here, $C_j$ is a single chore corresponding to clause $j$ of the SAT instance; and $0_i$ and $1_i$ are each a chore corresponding to literals of the binary variable $Y_i$ in the SAT instance.}, denoted by 
$$\{C_j:j\in [s]\}\cup \{x_i,1_i,0_i,1_i^*,0_i^*:i \in [t]\}$$

{\bf Preferences:} For any $i\in [t]$, let 
$${\mathcal C}_i^+=\{C_j: C_j\text{ includes }Y_i\}$$ 
and $${\mathcal C}_i^-=\{C_j: C_j\text{ includes }\neg Y_i\}$$
In other words, ${\mathcal C}_i^+$ (respectively, ${\mathcal C}_i^-$) is the set of clauses that are satisfied by $Y_i=1$ (respectively, $Y_i=0$). 

For every $i\in [t]$, define the importance orderings of agents $X_i,\neg X_i,X_i^*,\neg X_i^*$ as follows.
\newcommand{\others}{\text{others}}
\begin{align*}
\text{Agent }X_i : &\others \id 0_i\id x_i\id 1_i\\
\text{Agent }\neg X_i: &\others \id 1_i\id x_i\id 0_i\\
\text{Agent }X_i^*: &\others \id x_i^*\id {\mathcal C}_i^+ \id 1_i \id 1_i^*\\
\text{Agent }\neg X_i^*: &\others \id x_i\id {\mathcal C}_i^-\id0_i \id 0_i^*,
\end{align*}
where alternatives in ``$\others$'', $\mathcal C_i^+$, and $\mathcal C_i^-$ are ranked w.r.t.~a fixed order, e.g.~alphabetical order. In particular, $X_i$ and $\neg X_i$ have exactly the same importance ordering after the third chore from the bottom.

($\Rightarrow$) Suppose the {\sc SAT} instance has a solution, denoted by $y_1,\ldots, y_t$. Then we define the following allocation: for every $i\in [t]$, 
\begin{itemize}
\item if $y_i=1$, then agent $X_i $ gets $\{1_i,x_i\}$, agent $\neg X_i$ gets $\{0_i\}$, agent $X_i^*$ gets $1_i^*$ and as many unchosen chores in $\mathcal C_i^+$ as possible, and agent $\neg X_i^*$ gets $\{0_i^*\}$; 
\item  if $y_i=1$, then agent $X_i $ gets $\{1_i\}$, agent $\neg X_i$ gets $\{0_1,x_i\}$, agent $X_i^*$ gets $\{1_i^*\}$,  and agent $\neg X_i^*$ gets $0_i^*$ and as many unchosen chores in $\mathcal C_i^-$ as possible;
\end{itemize}
It is not hard to verify that the allocation is complete and EF. Notice that $X_i$ and $\neg X_i$ always get their (favorite) bottom-$1$ or bottom-$2$ chores, which means that they would not envy any other agent. Also notice that when $y_i=1$, agent $X_i^*$ does not envy agent $X_i$ because $X_i$ has chore $x_i$. Similarly, when $y_i=0$, agent $\neg X_i^*$ does not envy agent $\neg X_i$ because $\neg X_i$ has item $x_i$. 

($\Leftarrow$) Now suppose that the chores instance has a complete and envy-free allocation $A$. Let us define the assignment of values to the Boolean variables as follows. For every $i\in [t]$, if $x_i$ is assigned to agent $X_i$ (respectively, $\neg X_i$), then we let variable $Y_i=1$ (respectively, $Y_i=0$). If $x_i$ is not assigned to agent $X_i$ or $\neg X_i$, then the value of variable $Y_i$ is chosen arbitrarily---as we will see soon below, this case would not happen. 

We prove that $F$ is satisfied under this assignment. First, we prove that under any complete and EF allocation, for any $i\in [t]$, agent $X_i$ and $\neg X_i$ cannot be allocated a chore that is ranked strictly above $x_i$, i.e.~in $M\setminus \{1_i,x_i\}$ and $M\setminus \{0_i,x_i\}$, respectively. Suppose for the sake of contradiction this is not true, and suppose without loss of generality that agent $X_i$'s worst chore has the largest rank among all chores allocated to $X_i$ or $\neg X_i$. Then, since $X_i$ and $\neg X_i$ have identical importance orderings for items ranked above (i.e., worse than) $x_i$ in their lists, envy-freeness dictates that neither of them should get an item they consider worse than $x_i$.

Next, we prove that $x_i$ must be allocated to $X_i$ or $\neg X_i$. Suppose for the sake of contradiction that this is not true. Then the only possibility is that $x_i$ is allocated to $X_{i'}^*$ or $\neg X_{i'}^*$ for some $i'\in [t]$. If $x_i$ is allocated to $X_{i'}^*$, then $X_{i'}^*$   would envy $X_{i'}$.   More precisely, if $i=i'$, then $X_{i'}^*$  envies $X_{i'}$ because $X_{i'}$ only gets $\{1_i\}$, and $X_{i'}^*$ prefers $1_i$ to $x_i$. If $i\ne i'$, then $X_{i'}^*$ strictly prefers $x_{i'}$ to $x_i$, which means that $X_{i'}^*$ envies $X_{i'}$ regardless of the allocation of $x_i$. Similarly, if $x_i$ is allocated to $\neg X_{i'}^*$, then $\neg X_{i'}^*$   would envy $\neg  X_{i'}$. In either case there is a contradiction to EF. 

Therefore, for any $j\in [s]$, clause $C_j$ is not allocated to $X_i$ or $\neg X_i$ for any $i$. Then, we have the following two cases.
\begin{itemize}
\item Suppose $C_j$ is allocated to $X_i^*$. If $C_j$ is ranked above $x_i$ w.r.t.~agent $X_i^*$'s importance ordering, then agent $X_i^*$ would envy agent $X_i$ according to the reasoning above. Therefore, $C_j\in{\mathcal C_i^+}$, which means that agent $X_i^*$ prefers $1_i$ to $C_j$. Therefore, $x_i$ must be allocated to agent $X_i$, otherwise $X_i^*$ would envy $X_i$ (whose allocation is a single item $1_i$ in this case). This means that $C_j$ is satisfied under the assignment defined above. 
\item Suppose $C_j$ is allocated to $\neg X_i^*$, then following a similar line of reasoning we know that $C_j$ is satisfied by  $X_i=0$.
\end{itemize} 
In either case the clause is satisfied by the assignment defined above. This completes the proof of \Cref{thm:EF_Chores_NPC}.
\end{proof}

\subsection{Hardness of \EF{}+\PO{} for Chores}
\label{sec:EF+PO_Chores}

\begin{restatable}[\textbf{\EF{}+\PO{} for chores}]{cor}{EFPOChoresNPcomplete}
Determining whether a chores instance with lexicographic preferences admits an envy-free and Pareto optimal allocation is \NPC{}.
\label{cor:EF+PO_Chores_NPC}
\end{restatable}

\begin{proof}

We use an identical reduction to that of the proof of \Cref{thm:EF_Chores_NPC}. 

($\Rightarrow$) We know that given a SAT instance and a solution $y_1,\dots,y_t$, the  allocation in the proof of \Cref{thm:EF_Chores_NPC} is \EF{}. Here we show that the same allocation is the output of the sequence $\tau$ constructed below, and therefore \PO{} due to \Cref{prop:po_seqencible_chores}. Starting from an empty sequence $\tau$, we construct $\tau$ as follows. For each $i\in [t]$, \begin{itemize}
    \item if $y_i=1$, add to $\tau$ the subsequence $X_i,X_i,\neg X_i,X_i^*$, followed by $X_i^*$ as many times as the number of chores in $\mathcal{C}_i$ that $X_i^*$ gets in the allocation, followed by $\neg X_i^*$.
    \item if $y_i=0$, add to $\tau$ the subsequence $X_i, \neg X_i, \neg X_i, X_i^*, \neg X_i^*$, followed by $\neg X_i^*$ as many times as the number of chores in $\neg \mathcal{C}_i$ that $X_i^*$ gets in the allocation.
\end{itemize}
It is easy to verify that the picking sequence $\tau$ constructed above produces the allocation in the proof of \Cref{thm:EF_Chores_NPC}.

($\Leftarrow$) The converse follows from the proof of \Cref{thm:EF_Chores_NPC}.
\end{proof}

\section{Omitted Material from Section~\ref{subsec:EFX+PO} on \texorpdfstring{\EFX{}+\PO{}}{EFX+PO}}

\subsection{Proof of Theorem~\ref{thm:EFX_Mixed_TopItem}}

\EFXMixedTopItem*

\begin{proof}

Let $A$ be the output of Algorithm~\ref{alg:EFX+PO_mixed}. Let the total number of rounds be $R$. Since no agent can be selected more than once due to line~\ref{line:EFX+PO_mixed_eliminate_agent} of Algorithm \ref{alg:EFX+PO_mixed}, w.l.o.g. let agents be named according the round during the execution of Algorithm~\ref{alg:EFX+PO_mixed} in which they are selected to pick items. 

For each round $i\in\{1,2,\dots,R\}$, we define $N^i$ to be the remaining agents, $M^i$ to be the unallocated items, and $A^i$ to be the allocated items, at the {\em beginning} of round $i$. Specifically, $A^1=\emptyset$ and $M^1=M$, and for any $i\in\{2,\dots,R\}$, $A^i=\bigcup_{h\in\{1,\dots,i-1\}}A_h$ and $M^i=M\setminus A^i$. We refer to the instance $\langle N^i,M^i,G,C,\id\rangle$ as the reduced instance at the beginning of round $i$. For each agent $i\in\{1,2,\dots,R\}$, let $g_i^+$ be the top good in $A_i$, i.e., $g_i^+=\id_i(1,A_i\cap G_i)$. Before we begin, we introduce two lemmas about the structure of the output of Algorithm~\ref{alg:EFX+PO_mixed} which are used later in the proof.

\Cref{lem:EFX_PO_most_important_item} shows that for any agent $i\in\{1,2,\dots,R\}$, the most important item allocated to $i$ is a good.

\begin{lem}
For every agent $i\in\{1,2,\dots,R\}$, $g_i^+=\id_i(1,A_i)$.
\label{lem:EFX_PO_most_important_item}
\end{lem}
\begin{proof}
For agent $i$, suppose there exists an item $o\in A_{i}$ such that $o\id_{i} g_{i}^+$. 
\begin{itemize}[topsep=0pt]
    \item {\em If $o$ is a good for agent $i$}, it contradicts the assumption that $g_i^+$ is the top good in $A_i$.
    \item {\em Suppose $o$ is a chore for agent $i$.} Then, either:
    \begin{enumerate}[label={\em (Case \arabic*)},align=left,wide,topsep=0pt,labelindent=0pt]
        \item Item $o$ is a chore for every remaining agent in $N^{i}$ (at the beginning of round $i$). Then, $o$ must be in $A^{i}$ by lines~\ref{line:EFX+PO_mixed_find_chores}-\ref{line:EFX+PO_mixed_allocate_chores}, meaning $o$ is not available in $M^i$ to be picked.
        \item Otherwise, $o$ is a good for some agent $j\in N^i$, with $j>i$. Then, $o\not\in C'$ in line~\ref{line:EFX+PO_mixed_find_chores} at round $i$, which implies that $o\not\in A_{i}$.
    \end{enumerate}
\end{itemize}
\end{proof}

\Cref{lem:EFX_PO_goods_only} shows that after round $1$, every agent only picks up items which they perceive to be goods.

\begin{lem}
For every agent $i\in\{2,\dots,R\}$, $A_i\subseteq G_i$.
\label{lem:EFX_PO_goods_only}
\end{lem}
\begin{proof}
Suppose for the sake of contradiction an agent $i\in\{2,\dots,R\}$ gets an item $c$ it considers a chore, i.e., $c\in A_i\cap C_i$. By \Cref{lem:EFX_PO_most_important_item} and line~\ref{line:EFX+PO_mixed_allocate_chores}, $c$ must be a chore for every agent $j\in\{i+1,\dots,R\}$. This implies that $c$ is a chore for every agent in $\{i,\dots,R\}$. However, by lines~\ref{line:EFX+PO_mixed_find_chores}-\ref{line:EFX+PO_mixed_allocate_chores}, this is impossible and $c$ must have been picked up by an agent $h<i$ in an earlier round, a contradiction.
\end{proof}

\paragraph{Part 1: Algorithm~\ref{alg:EFX+PO_mixed} satisfies \EFX{}.} Suppose for the sake of contradiction that an agent $i\in N$ envies another agent.

\paragraph{Part 1.1: Agent $i$ does not envy a later agent $j>i$.} If $i=1$: agent $1$ gets its top-good which is also its top-item, and therefore does not envy any other agent.

Otherwise, if $i\in\{2,\dots,R\}$: Suppose for the sake of contradiction, agent $i$ envies an agent $j>i$, i.e., $A_j\succ_i A_i$. Any good $g\in G_i$ such that $g\id_i g_i^+$, must be already allocated to an earlier agent $h<i$, by construction, since $i$ picks its top remaining good in round $i$. Therefore, $g$ cannot be a source of envy towards any agent $j>i$ by our assumption of lexicographic preferences, and $A_i\succ_i A_j$ for any $j>i$, a contradiction.

\paragraph{Part 1.2: Agent $i$'s envy for an earlier agent $h<i$ can be eliminated by the removal of any good from $A_h$.} Suppose agent $i$ envies an agent $h<i$. By construction, $g_h^+$ is the only item in $A_h$ that $i$ considers a good, and every other item is a chore for agent $i$, i.e., $A_h\setminus\{g_h^+\}\subseteq C_i$. Then, by \Cref{lem:EFX_PO_goods_only} which implies that $A_i\subseteq G_i$, it holds that $A_i\succeq A_h\setminus\{g_h^+\}$ due to lexicographic preferences.

Together, this implies $A$ is \EFX{}, meaning Algorithm~\ref{alg:EFX+PO_mixed} satisfies \EFX{}.

\paragraph{Part 2: Algorithm~\ref{alg:EFX+PO_mixed} satisfies \PO{}.}
The proof proceeds in the following steps:
\begin{itemize}
    \item (\Cref{lem:EFX_PO_most_important_good})
    We show that for every agent $i\in\{1,2,\dots,R\}$, agent $i$ must retain $g_i^+$ in any Pareto dominating allocation $B$. This is because, if agent $i$ were to lose $g_i^+$ in $B$:
    \begin{itemize}
        \item If $i=1$: Agent $1$ cannot be compensated since $g_1^+$ is also its top item overall (by construction). 
        
        If $i\in\{2,\dots,R\}$: Agent $i$ cannot be compensated for losing the good $g_i^+$, which is also its top-ranked item in its allocation (\Cref{lem:EFX_PO_most_important_item}),  by getting a good from $A_j$ for any agent $j>i$ who was allocated items in a later round, since $i$ picks its top remaining good in $M^i$ in round $i$.
        \item Agent $i$ cannot be compensated by gaining a good from $A_h$ for any earlier agent $h<i$ by induction. If agent $h$ cannot lose $g_h^+$, every other item must be a chore for agent $i$ by construction.
        \item Agent $i$ also cannot be compensated by losing a chore: Agent $1$ cannot be compensated since $g_1^+$ is also its top item overall. Every other agent only receives goods (\Cref{lem:EFX_PO_goods_only}).
    \end{itemize} 
    \item (\Cref{lem:EFX_PO_other_items}) We show that for every agent $i\in\{R,R-1,\dots,1\}$, and every item $o\in A_i\setminus\{g_i^+\}$, agent $i$ must retain $o$ in $B_i$ by induction:
    \begin{itemize}
        \item For agents $i\in\{R,R-1,\dots,2\}$, $o$ must be a good due to \Cref{lem:EFX_PO_goods_only}. Therefore, this loss of a good must be compensated by adding another good.
        \item By the induction hypothesis, the good cannot come from an agent $j>i$. Due to \Cref{lem:EFX_PO_most_important_good}, no agent $h<i$ can lose $g_h^+$. The only other items to consider must be chores to agent $i$, meaning that $i$ cannot be compensated for losing the good $o$.
        \item If agent $1$ loses a chore, and it goes to an agent $j>1$, then by the induction hypothesis and \Cref{lem:EFX_PO_most_important_good}, $j$ is strictly worse off since all agents $k>1$ retain all the items they get in $A$.
    \end{itemize}
\end{itemize}

Suppose for the sake of contradiction that there exists an assignment $B$ which Pareto dominates $A$.

\paragraph{Part 2.1: Every agent $i$ retains its top good $g_i^+$.}

\begin{lem}
For every agent $i\in\{1,2,\dots,R\}$, $g_i^+\in B_i$.
\label{lem:EFX_PO_most_important_good}
\end{lem}
\begin{proof}
\noindent{\em High level idea.} We show by induction that if $i$ loses $g_i^+$, then agent $i$ is worse off in $B$, i.e., $A_i\succ_i B_i$, a contradiction to our assumption that $B$ Pareto dominates $A$. This relies on \Cref{lem:EFX_PO_most_important_item} that the good $g_i^+$ is $i$'s top item in $A_i$. 

We show that if $i$ loses $g_i^+$:
(1) Adding another item in $M^i$ which is available at round $i$ does not compensate agent $i$ for losing $g_i^+$. Further, we show that agent $i$ cannot be compensated by adding an item in $A^i$ either. This is because for any agent $h<i$ who picked items in a previous round, our induction hypothesis means every previous agent $h$ retains $g_h^+$ in $B_h$, and by lines~\ref{line:EFX+PO_mixed_find_chores}-\ref{line:EFX+PO_mixed_allocate_chores} every other item in $A_h$ is a chore for agent $i$.

\vspace{0.5em}
\noindent{\bf Base case.} True for $i=1$ due to the construction (lines~\ref{line:EFX+PO_mixed_first_agent}-\ref{line:EFX+PO_mixed_end_first_agent}), $g_1^+$ is agent $1$'s top item among all items. If agent $1$ loses $g_1^+$, due to lexicographic preferences, $A_1\succ_1 B_1$, a contradiction to our assumption that $B$ Pareto dominates $A$. 

\vspace{0.5em}
\noindent{\bf Induction step.} Suppose it holds for every $i\in\{1,\dots,r\}$, where $r<R$, that $g_i^+\in B_i$. Consider agent $r+1$.

\vspace{0.5em}
\noindent(1) {\em Agent $r+1$ cannot be compensated by an item in $M^{r+1}$.} Since $g_{r+1}^+$ is agent $r+1$'s most important item (\Cref{lem:EFX_PO_most_important_item}), removing $g_{r+1}^+$ and adding an item $o\in M^{r+1}\setminus\{g_{r+1}^+\}$, cannot result in a Pareto improvement: If $o$ is a good, it must be ranked higher than $g_{r+1}^+$ by $r+1$. Then, if $o$ remains in $M^{r+1}$, $r+1$ must be assigned $o$ in line~\ref{line:EFX+PO_mixed_allocate_chores}, which contradicts $g_{r+1}^+$ being agent $r+1$'s top good in $A_{r+1}$. Otherwise, if $o$ is a chore, adding it does not compensate for losing $g_{r+1}^+$.

\vspace{0.5em}
\noindent(2) {\em Agent $r+1$ cannot be compensated by an item in $A^{r+1}$.} Every item $o\in A^{r+1}$ must have been picked in a strictly earlier round which means there is some $h<r+1$, such that $o\in A_h$.
\begin{itemize}[topsep=0pt]
    \item {\em If $o$ is the top item $g_h^+$ for $h$}, then by our assumption that the induction hypothesis is true for all $i\in\{1,\dots,r\}$, it is retained by agent $h$ in $B_h$ and cannot be in $B_{r+1}$.
    \item {\em Otherwise, $o$ must be a chore for agent $r+1$} by lines~\ref{line:EFX+PO_mixed_find_chores}-\ref{line:EFX+PO_mixed_allocate_chores}. Therefore, adding $o$ to $B_{r+1}$ does not result in a Pareto improvement due to lexicographic preferences. 
\end{itemize}

\vspace{0.5em}
\noindent(3) {\em Agent $r+1$ cannot be compensated by losing a chore.} By \Cref{lem:EFX_PO_goods_only}, $r+1$ does not get any chore.

Therefore since the good $g_{r+1}^+$ is $r+1$'s top item in $A_i$ (\Cref{lem:EFX_PO_most_important_item}) and $r+1$ cannot obtain a good $o\id_{r+1} g_{r+1}^+$, losing $g_{r+1}^+$ implies $A_{r+1}\succ_{r+1}B_{r+1}$, which contradicts our assumption that $B$ Pareto dominates $A$.
\end{proof}

\paragraph{Part 2.2: Every agent $i$ retains every item in $A_i\setminus\{g_i^+\}$.}

\begin{lem}
For every agent $i\in\{1,2,\dots,R\}$, if $o\in A_i\setminus\{g_i^+\}$, then $o\in B_i$.
\label{lem:EFX_PO_other_items}
\end{lem}
\begin{proof}

By induction, starting with the last agent to pick items, we show that (1) an agent can only lose goods but not chores, and (2) that the agent cannot be compensated for losing the good. 

\noindent{\bf Base case.} For $i=R$, suppose for the sake of contradiction $o\in A_R\setminus B_R$ is an item agent $R$ loses. 

\vspace{0.5em}
\noindent(1) Item $o$ must be a good, i.e., $o\in G_R$ due to \Cref{lem:EFX_PO_goods_only}.

\vspace{0.5em}
\noindent(2) Suppose $o'\in (B_R\setminus A_R)\cap G_R$ is the good gained by agent $R$ to compensate. Then, $o'\in A_h$ for some agent $h<R$.
\begin{itemize}
    \item If $o'$ is agent $h$'s top good in $A_h$, i.e., $g_h^+$, then due to \Cref{lem:EFX_PO_most_important_good}, $o'\in B_h$, a contradiction. 
    \item Otherwise, since $R$ is the last agent to pick items and $o'$ is a good for agent $R$, $o'$ cannot have been picked by agent $h$ due to lines~\ref{line:EFX+PO_mixed_find_chores}-\ref{line:EFX+PO_mixed_allocate_chores}, a contradiction.
\end{itemize}

\noindent{\bf Induction step.} Suppose it holds true for every $i\in\{R,R-1\dots,r\}$, where $r>1$, that if $o\in A_i\setminus\{g_i^+\}$, then $o\in B_i$. Consider the case of agent $r-1$.

\vspace{0.5em}
\noindent{\em The special case of $r-1=1$.} Every item $o\in A_1\setminus\{g_1^+\}$ is a chore for every other agent by construction (lines~\ref{line:EFX+PO_mixed_first_agent}-\ref{line:EFX+PO_mixed_end_first_agent}). By our induction assumption and due to \Cref{lem:EFX_PO_most_important_good}, giving $o$ to an agent $j>1$ who considers $o$ to be a chore results in $j$ being strictly worse off. This contradicts our assumption that $B$ Pareto dominates $A$.

\vspace{0.5em}
\noindent{\em Now, suppose $r-1>1$.}
\begin{enumerate}[label=(\arabic*)]
    \item Item $o$ must be a good due to \Cref{lem:EFX_PO_goods_only}.
    \item Then, $r-1$ must gain a good to compensate (since $r-1$ cannot lose a chore) to compensate for losing the good $o$.
    
    Suppose item $o'\in B_{r-1}\setminus A_{r-1}$ is an item agent $r-1$ gains. By our induction assumption and due to \Cref{lem:EFX_PO_most_important_good}, $o'$ must belong to some agent $h<r-1$, i.e., $o'\in A_h$ and it cannot be $g_h^+$, i.e., $o'\in A_h\setminus\{g_h^+\}$. Then, by lines~\ref{line:EFX+PO_mixed_find_chores}-\ref{line:EFX+PO_mixed_allocate_chores}, $o'$ must be a chore for agent $r-1$. Therefore, $o'$ cannot be compensated for the loss of good $o$, a contradiction to our assumption that $B$ Pareto dominates $A$.
\end{enumerate}

This completes the proof.
\end{proof}

Together, Lemmas~\ref{lem:EFX_PO_most_important_good} and~\ref{lem:EFX_PO_other_items} above imply that for every agent $i$, $B_i=A_i$, which means $B$ does not Pareto dominate $A$, a contradiction.
\end{proof}

\subsection{Proof of~\cref{cor:EFX_Mixed_special}}
\EFXMixedspecial*
\begin{proof}

Consider a lexicographic mixed instance where every item is considered a good by at least one agent. If there exists an agent $i$ with the top-ranked item as a good, i.e., $\id_i(1)\in G_i$, then Algorithm \ref{alg:EFX+PO_mixed} guarantees an \EFX{} and \PO{} allocation according to \cref{thm:EFX_Mixed_TopItem}. Therefore, for the rest of the proof, we will focus on the case where every agent's top-ranked item is a chore, i.e., $\id_i(1) \in C_i$ for all $i\in N$.

Consider a variation of Algorithm \ref{alg:EFX+PO_mixed} that starts from line \ref{line:no_common_chores} (i.e. where there are no common chores). We show that this variation satisfies \EFX{} and \PO{} when every item is considered a good by at least one agent. 

Given an arbitrary profile $\id$, let $A$ be the output of the variation of Algorithm \ref{alg:EFX+PO_mixed} above, and $R$ be the total number of rounds of execution of Algorithm \ref{alg:EFX+PO_mixed} on profile $\id$. For each round $i\in\{1,2,\dots,R\}$, we define $N^i$ to be the remaining agents, $M^i$ to be the unallocated items, and $A^i$ to be the allocated items, at the {\em beginning} of round $i$. Specifically, $A^1=\emptyset$ and $M^1=M$, and for any $i>1$, $A^i=\bigcup_{h=1,\dots,i-1}A_h$ and $M^i=M\setminus A^i$. We refer to the instance $\langle N^i,M^i,G,C,\id\rangle$ as the reduced instance at the beginning of round $i$. Since no agent can be selected more than once due to line~\ref{line:EFX+PO_mixed_eliminate_agent} of Algorithm \ref{alg:EFX+PO_mixed}, w.l.o.g., for any $i\in\{1,2,\dots,R\}$, let agent $i$ be the agent selected by Algorithm \ref{alg:EFX+PO_mixed} at round $i$ of the algorithm. Let $g_i^+$ be the top good in $A_i$.

We begin by showing in \cref{lem:EFX_PO_Mixed_special_top_item_top_good} that at any round $i=1,\dots,R$, agent $i$'s most important item in $A_i$ is the good $g_i^+$, using a similar argument to \Cref{lem:EFX_PO_most_important_item}.

\begin{lem}
For every agent $i\in\{1,2,\dots,R\}$, $g_i^+=\id_i(1,A_i)$.
\label{lem:EFX_PO_Mixed_special_top_item_top_good}
\end{lem}
\begin{proof}
Suppose for the sake of contradiction that agent $i$'s top item in $A_i$ is a chore $o\in C_i$, i.e., $\id_{j}(1,A_i)\in C_i$. If $o$ is also a chore for every other remaining agent in $N^i$, then according to lines~\ref{line:EFX+PO_mixed_find_chores}-\ref{line:EFX+PO_mixed_allocate_chores} of Algorithm \ref{alg:EFX+PO_mixed}, $o$ must have been allocated to an agent in a strictly earlier round $h<i$, a contradiction. Otherwise, if $o$ is a good for some remaining agent $j\in N^i$, where $j>i$, then according to lines~\ref{line:EFX+PO_mixed_find_chores}-\ref{line:EFX+PO_mixed_allocate_chores} of Algorithm \ref{alg:EFX+PO_mixed}, agent $i$ cannot be allocated $o$, a contradiction.
\end{proof}

\paragraph{Part 1: The variation of Algorithm~\ref{alg:EFX+PO_mixed} satisfies \EFX{}.}
Consider any pair of agents $i$ and $j$. Without loss of generality, assume that $i$ was selected by the algorithm before $j$. 

If agent $j$ is envious of $A_i$, the envy must be towards a single good $o\in A_i$ since all other items allocated to $i$ are common chores by every agent selected after agent $i$. Thus, agent $j$'s envy, if any, can be eliminated by removing a single good from $A_i$.

Moreover, agent $i$ cannot be envious of $j$ because when $i$ is selected prior to $j$ by the algorithm, $i$'s top good $g_i^+\in A_i$ is more important than any item in $A_j$, i.e., $g_i^+ \>_{i}\ A_{j}$. Since $g_i^+$ is also the most important item in $A_i$ due to \cref{lem:EFX_PO_Mixed_special_top_item_top_good}, i.e., $\id_i(1,A_i)=g_i^+$, by our assumption of lexicographic preferences, agent $i$ does not envy agent $j$.

Finally, if there are any excess items, the last agent receives them all. These items must be goods for the last agent since all other chores are assigned in previous rounds. Since preferences are lexicographic, every agent selected before the last agent prefers its own bundle to that of the last agent.

\paragraph{Part 2: The variation of Algorithm~\ref{alg:EFX+PO_mixed} satisfies \PO{}.}

Suppose there exists an allocation $B$ that Pareto dominates $A$.

We begin by showing that for agent $1$, $g_1^+\in B_1$. Suppose $g_1^+\not\in B$. Since $B$ Pareto dominates $A$, this implies that $B_1\>_1 A_1$. Then, it must hold that either (i) there is a good $g\in G_1$ such that $g\in B$ and $g\id_1 g_1^+$, which is impossible by construction and the choice of agent $1$, or (ii) there is a chore $c\in C_1\cap A_1$ such that $c\not\in B$ and $c\id_1 g_1^+$, which is impossible since $g_1^+$ is the most important item in $A_1$ due to \cref{lem:EFX_PO_Mixed_special_top_item_top_good}.

By a similar argument to Part 2.1 of the proof of Pareto optimality of Algorithm~\ref{alg:EFX+PO_mixed} in \cref{thm:EFX_Mixed_TopItem}, it is easy to see that for every agent $i=2,\dots,R$, $g_i^+\in B$.

The rest of the proof follows a similar argument to Part 2.2 of the proof of Pareto optimality of Algorithm~\ref{alg:EFX+PO_mixed} in \cref{thm:EFX_Mixed_TopItem} to show that in $B$, every agent $1,\dots,R$ must retain every item in $A_i\setminus\{g_i^+\}$ as well. This leads us to conclude that for every agent $B_i=A_i$ which implies that $B$ does not Pareto dominate $A$, a contradiction.
\end{proof}

\section{Omitted Material from Section~\ref{subsec:MMS} on \texorpdfstring{\MMS{}}{MMS}}

\subsection{Proof of Proposition~\ref{prop:MMS_mixed_char}}

\MMSmixed*
\begin{proof}
Consider the agent $i\in N$. Let $\id_i$ denote the importance ranking of agent $i$ over the items (goods and chores) in $G_i \cup C_i$. The \MMS{} partition of agent $i$ is uniquely defined based on its top-ranked item $\id_i(1)$ in the following cases:

\begin{itemize}
    \item \textbf{Top ranked item is a good, i.e., $\id_i(1) \in G_i$}: There are two cases according to the size of $G_{i}$.
    
    \begin{enumerate}
        \item If $|G_{i}| \geq n$: the \MMS{} partition for $i$ is defined as 
        $$\footnotesize \{ \{\id_{i}(1, G_i) \cup C_i\}, \{\id_{i}(2, G_i)\}, \ldots, \{\id_{i}([n, m], G_i)\} \}.$$
        The \MMS{} partition for $i$ is the least preferred bundle. since preferences are lexicographic, we have $\MMS_i{} = \{\id_{i}([n, m], G_i)\} = G_i \setminus \{ \id_i([n-1], G_i) \}$.

        \item If $|G_{i}| < n$: the \MMS{} partition for $i$ is uniquely defined as 
        \begin{equation*}
        \begin{split}
        \{ \{\id_{i}(1, G_i) \cup C_i\}, \{\id_{i}(2, G_i)\}, \ldots, \{\id_{i}(|G_i|, G_i)\},
        \{\}, \ldots, \{\} \}.
        \end{split}
        \end{equation*}
        Therefore, $\MMS_i{} = \emptyset$.
    \end{enumerate}

    \item \textbf{Top ranked item is a chore, i.e., $\id_i(1) \in C_i$}: The \MMS{} partition is uniquely defined as
    $$\{ \{\id_i(1, C_i)\cup G_i \}, \{\id_i(2, C_i)\}, \ldots, \{\id_i(n, C_i)\} \}.$$
    
    Not that if $|C_i| < n$, then $\{\id_i(k, C_i)\} = \emptyset$ for all $k < |C_i|$. The \MMS{} for agent $i$ is the least preferred partition above. Since preferences are lexicographic, $\MMS_i{} = \{\id_i(1, C_i)\cup G_i \}$.
\end{itemize}
\end{proof}

\subsection{Proof of \cref{prop:EFX_MMS_mixed}}

\EFXImpliesMMSForMixed*

\begin{proof}
Let $A$ be an \EFX{} allocation. Suppose, for contradiction, that $A$ does not satisfy \MMS{}, and say the \MMS{} guarantee is violated for agent $i$. Then, using the characterization of \MMS{} in \Cref{prop:MMS_mixed_char}, we have the following cases depending on whether agent $i$'s top-ranked item, namely $\id_{i}(1)$, is a good for it, and whether the number of items it considers to be goods, namely $|G_i|$, is at least $n$:
\begin{enumerate}
	\item \emph{When agent $i$ considers its top-ranked item a good (i.e.,~$\id_i(1) \in G_i$) and there are at least $n$ items that it considers goods~(i.e.,~$|G_{i}| \geq n$)}.\\
	
	In this case, the \MMS{} share of agent $i$ is the set of all goods for agent $i$ except for its favorite $(n-1)$ goods. We denote this set as $G_i^{n+} \coloneqq G_i \setminus \{\id_{i}([n-1], G_i)\}$. Thus, a violation of \MMS{} requires that	
	$$G_i^{n+} \, \>_i \, A_i.$$
	
	In particular, agent $i$ does not receive any of its favorite $(n-1)$ goods under $A$.
	
	We will consider two subcases, depending on whether the bundle $A_i$ contains a chore for agent $i$.\\
	
	\begin{enumerate}
		\item \emph{When $A_i \cap C_i \neq \emptyset$}.\\
		
		Let $h$ denote an agent who owns agent $i$'s top-ranked item, i.e., $\id_i(1) \in A_h$. Then, because of lexicographic preferences, $i$ must envy $h$ under $A$. Furthermore, $i$ continues to envy $h$ even after the removal of some chore in $A_i$, thus violating \EFX{}.\\
		
		\item \emph{When $A_i \cap C_i = \emptyset$}.\\
		
		Since agent $i$ does not receive any of its chores and since $A$ is not \MMS{} for $i$, there must exist an item $o' \in G_i^{n+}$ such that $o' \notin A_i$. Recall that agent $i$ does not receive any of its $(n-1)$ favorite goods in $G_i$, say $o_1,\dots,o_{n-1}$. Thus, there are $n$ items, namely $o_1,\dots,o_{n-1}$ and $o'$, that are allocated among $(n-1)$ other agents. Let $h$ be an agent who receives at least two of these $n$ items. Then, $h$ must receive at least one of agent $i$'s favorite $(n-1)$ goods, and thus agent $i$ envies agent $h$ even after the removal of the other item in $h$'s bundle, implying a violation of \EFX{}.\\
	\end{enumerate}	

		\item \emph{When agent $i$ considers its top-ranked item a good (i.e.,~$\id_i(1) \in G_i$) and there are fewer than $n$ items that it considers goods~(i.e.,~$|G_{i}| < n$)}.\\
	
	In this case, the \MMS{} share of agent $i$ is the empty set. Thus, a violation of \MMS{} requires that
	$$\emptyset \, \>_i \, A_i.$$

	In particular, this implies that agent $i$'s top-ranked item in its bundle must be a chore (i.e., $\id_i(1,A_i) \in C_i)$, and that agent $i$ does not receive its top-ranked item which is a good for it (i.e., $\id_i(1) \notin A_i$).
	
	Let $h$ be the agent who owns agent $i$'s top-ranked item under $A$, i.e., $\id_i(1) \in A_h$. Then, agent $i$ envies agent $h$ even after the removal of the top-ranked chore in its bundle, namely $\id_i(1,A_i)$, which contradicts \EFX{}.\\
	
	\item \emph{When agent $i$ considers its top-ranked item a chore (i.e.,~$\id_i(1) \in C_i$)}.\\
	
	In this case, the \MMS{} share of agent $i$ is its top-ranked chore along with all the items it considers goods, namely $\id_i(1) \cup G_i$. Thus, a violation of \MMS{} requires that
	$$\id_i(1) \cup G_i \, \>_i \,  A_i.$$
	In particular, agent $i$'s bundle $A_i$ must contain its most undesirable chore $\id_i(1)$ and a strict subset of its goods $G_i$.
	Let $o \in G_i$ be a good for agent $i$ that is not in its bundle $A_i$, and let $h$ be the agent who own $o$ under $A$, i.e., $o \in A_h$. Then, even after removing $o$ from $h$'s bundle, agent $i$ prefers $h$'s bundle over its own, implying a violation of \EFX{}.
\end{enumerate}

Since we obtain a contradiction in each case, the \EFX{} allocation, $A$, must also be \MMS{}.

To see why \MMS{} does not imply \EFX{}, one can refer to \cref{prop:EFX_implies_MMS} for chores-only instances.
\end{proof}

\subsection{Computing \MMS{} for Mixed Items}

Algorithm~\ref{alg:MMS_mixed} provides a detailed description of how \MMS{} allocations can be computed for (even subjective) mixed items.

\begin{algorithm}[t]
\DontPrintSemicolon
 \linespread{1}
\KwIn{A lexicographic mixed instance $\langle N,M,G,C,\id \rangle$}
\KwOut{An \MMS{} allocation $A$}
Let $C' \coloneqq \{o\in M:\ \forall i\in N,\ o \in C_i \}$ \;
\Comment{{Step 1: Assigning chores according to top-ranked items}}

\If{$\exists i\in N$ such that $\id_i(1) \in G_i$}{
Run Algorithm~\ref{alg:EFX+PO_mixed}\;
}
\Else(\tcp*[h]{Else if $\forall i \in N,  \id_i(1) \in C_i$}){
Fix a priority ordering $\sigma$ over $n$ agents \;
\If{$|C'| \geq n$}{
Run a serial dictatorship where $\sigma_1$ picks his best $|C'|-n+1$ chores\;
All remaining agents pick one chore \;
}
\Else{
Agents pick one chore according to $\sigma$, and none if no chore is remaining \;
}
If exists an agent who picked its worst chore (highest priority), give that agent its remaining goods \;\label{line:MMS_mixed_allocate_goods}
}
\Comment{{Step 2: Serial dictatorship for assigning remaining items}}
Run a serial dictatorship according to priority ordering $\tau$; agents pick any number of goods among remaining items or nothing (if no item is a good for them). \;\label{line:MMS_mixed_step_2}
\KwRet{$A$}
\caption {Algorithm for finding an \MMS{} allocation for mixed items.}
\label{alg:MMS_mixed}
\end{algorithm}

\MMSPOMixedItems*

\begin{proof}
Algorithm~\ref{alg:MMS_mixed} guarantees \MMS{} for \emph{any} lexicographic instance with mixed items. Let $A$ be the allocation after running the algorithm.
The proof is based on two types of instances that are defined through the existence of a \emph{special} agent. An agent $i$ is said to be special if either (i) its top-ranked item is a good, i.e., $\id_i(1) \in G_i$ or (ii) receives its least-preferred chore, i.e., $\exists\ c \in A_i^{i-}$ such that $c\ \>_i\ c', \forall c' \in C_i \setminus \{c\}$.
Let $C' \coloneqq \{o\in M:\ \forall i\in N,\ o \in C_i \}$ be the set of common chores.

\textbf{Case (i)}: there exists an agent $i$ with top-ranked item as a good, that is, $\id_i(1) \in G_i$.
Then, run  Algorithm~\ref{alg:EFX+PO_mixed} that satisfies \EFX{} (and \PO{}). By \cref{prop:EFX_MMS_mixed}, any \EFX{} allocation is also \MMS{}, thus, Algorithm~\ref{alg:MMS_mixed} is \MMS{}.
Note that in this case, the algorithm does not allocate any item in `Step 2', thus, the allocation vacuously remains \MMS{}.

\textbf{Case (ii)}: every agents' top-ranked item is a chore, that is, $\forall i \in N,\ \id_i(1) \in C_i$.
The proof relies on allocating items that are considered as chores by all agents, i.e., $C' \coloneqq \{o\in M:\ \forall i\in N,\ o \in C_i \}$. All remaining items in $M\setminus C'$ by construction are considered goods by at least one agent.
Algorithm~\ref{alg:MMS_mixed} proceeds to first allocate items in $C'$---via a serial dictatorship specified by $\sigma$---such that the first agent $\sigma_1$ either receives its most preferred chore (if $|C'| < n$) or its most preferred $|C'| - n + 1$ chores (if $|C'| \geq n)$. All other agents pick a single chore from $C' \setminus A_{\sigma_1}$ or an empty set, which satisfies  \MMS{}.

Suppose agent $h$ receives its $n$th ranked-chore in $C_h$. Then, agent $h$ receives all goods in $G_{h}$ (Line~\ref{line:MMS_mixed_allocate_goods}). These items must be available since the set $C'$ did not contain any item that is considered good by any agent.
Thus, we have $A_{h}\ \succeq_h\ \id_{i}(1, C_h) \cup G_h$, which implies \MMS{} by \cref{prop:MMS_mixed_char}.
The allocation of remaining items only improves the outcome for all agents since all remaining items are assigned as goods in a serial dictatorship by Algorithm~\ref{alg:MMS_mixed} (Line~\ref{line:MMS_mixed_step_2}). Thus, all agents' allocations weakly improves.
\end{proof}

\section{Chores: Additional Results}
\label{sec:Chores}

In this section, we provide additional results for chores-only instances under lexicographic preferences.
Note that by our convention for chores-only instances, the top-ranked item in the importance ordering is actually the least-preferred chore.

\subsection{Envy-Freeness}

In \cref{app:EF} we addressed the question of computing an envy-free and \PO{} allocation. 
Let us now consider envy-freeness alongside rank-maximality. For goods, there is a polynomial-time algorithm known for determining whether an envy-free and rank-maximal allocation exists \cite{HSV+21fair}. 
By contrast, the problem turns out to be \NPC{} for chores (\Cref{thm:EF_RM_NP-complete_Chores}).

Notice that the weaker combination of \EF{1} and \PO{} is implied by~\Cref{prop:EFX+PO_Chores} since any \EFX{} allocation is, by definition, also \EF{1}.
Thus, from the viewpoint of exact envy-freeness, the chores problem seems to be significantly harder compared to goods.

\begin{restatable}[\textbf{\EF{}+\RM{} for chores}]{thm}{EFRMChores}
Determining whether a chores instance with lexicographic preferences admits an envy-free and rank-maximal allocation is \NPC{}.
\label{thm:EF_RM_NP-complete_Chores}
\end{restatable}

\begin{proof}
Membership in \NP{} follows from the fact that both envy-freeness and rank-maximality can be checked in polynomial time. To prove \NPH{}ness, we will show a reduction from \HRTC{}~\cite{DRS05hardness,GS17hardness} which asks the following question: Given a hypergraph $H = (V,E)$ with vertex set $V$ and edge set $E \subseteq 2^V$, does there exists a way of assigning each vertex exactly one of three given colors such that each hyperedge receives all three colors? This problem is known to be \NPC{}~\cite{KLL20np}.

\emph{Construction of the reduced instance}:
Let $q \coloneqq |V|$ and $r \coloneqq |E|$ denote the number of vertices and hyperedges in the given hypergraph $H = (V,E)$, respectively. 
We will construct a fair division instance with $n = r + 3$ agents and $m = (q+1)r+q+3$ chores. The set of agents consists of $r$ \emph{edge} agents $e_1,\dots,e_r$, and three \emph{dummy} agents $d_1,d_2,d_3$. The set of chores consists of $q \cdot r$ \emph{signature} chores $\{S_i^1,\dots,S_i^q\}_{i \in [r]}$, $r$ \emph{edge} chores $E_1,\dots,E_r$, $q$ \emph{vertex} chores $V_1,\dots,V_q$, and three \emph{dummy} chores $D_1,D_2,D_3$. With some notational overloading, we will write $E \coloneqq \{E_1,\dots,E_r\}$, $V \coloneqq \{V_1,\dots,V_q\}$, $S_i \coloneqq \{S_i^1,\dots,S_i^q\}$, and $D \coloneqq \{D_1,D_2,D_3\}$ to denote the unordered sets of edge, vertex, signature, and dummy chores, respectively. Further, we will write $S_{-i} \coloneqq \{S_1,S_2,\dots,S_{i-1},S_{i+1},\dots,S_r\}$ to denote the unordered set of signature chores after removing those in $S_i$, and write $V_{e_i}$ to denote the unordered set of vertex chores corresponding to the vertices in the hyperedge $e_i \in E$. Note that for all $i \in [r]$, $|V_{e_i}| \leq q$ and $|S_i| = q$.

\begin{table}[ht]
\footnotesize
\renewcommand{\arraystretch}{1.35}
\centering
    \begin{tabular}{|rl|}
    \hline
    $e_i$: & $V_{e_i} \id E \setminus \{E_i\} \id E_i \id D  \id S_{-i} \id V \setminus V_{e_i} \id S_i$
    \\
    \hline
    $d_\ell$: & $E_r \id \dots \id E_1 \id S_r \id \dots \id S_1 \id D \setminus D_\ell \id D_\ell \id V$ \\
    \hline
    \end{tabular}
    \caption{Importance orderings of agents in the proof of \Cref{thm:EF_RM_NP-complete_Chores}.}
    \label{tab:EF_RM_NP-complete_Chores}
\end{table}

\emph{Preferences}: \Cref{tab:EF_RM_NP-complete_Chores} shows the importance orderings of the agents in terms of a partial order over sets of items. For any pair of sets of chores $C$ and $C'$, we use $C\id C'$ to denote that all chores in $C$ are ranked higher than all chores in $C'$ in the importance ordering. To obtain the exact preference ordering, the partial orders can be extended by ranking the vertex chores in order of their indices (i.e., $V_q \id V_{q-1} \id \dots \id V_1$) and ranking the edge, signature, and dummy chores arbitrarily.

For every $i \in [r]$, the edge agent $e_i$ prefers the signature chores in the set $S_i$ over all other chores. Subject to that, all vertex chores except for those in $V_{e_i}$ are preferred over the remaining chores; the chores in $V_{e_i}$ are the least preferred chores of $e_i$. The remaining signature chores $S_{-i}$ are preferred over the dummy chores, which are then preferred over the edge chore $E_i$ and the remaining edge chores.

For any $\ell \in [3]$, the dummy agent $d_\ell$ prefers all vertex chores over the dummy chore $D_\ell$, which is preferred over the other two dummy chores and the signature chores, which, in turn, are preferred over the edge chores. This completes the construction of the reduced instance.

Notice that for every $i \in [r]$, any signature chore in the set $S_i$ realizes its largest ranking in the importance list of agent $e_i$. Therefore, any rank-maximal allocation must assign all chores in $S_i$ to $e_i$. By a similar reasoning, the edge chore $E_i$ is also assigned to $e_i$ in any rank-maximal allocation. Next, note that for any fixed $i \in [q]$, the vertex chore $V_i$ is ranked lower by the dummy agents than by the edge agents, and therefore must be assigned to one of $d_1$, $d_2$, or $d_3$ in any rank-maximal allocation. Finally, rank-maximality requires that for any $\ell \in [3]$, the dummy chore $D_\ell$ is assigned to the dummy agent $d_\ell$. It is straightforward to check that the aforementioned necessary conditions for rank-maximality are also sufficient.

We will now argue the equivalence of solutions.

($\Rightarrow$) Suppose the hypergraph $H$ admits a feasible rainbow $3$-coloring. Then, the desired allocation can be constructed as follows: For every $i \in [r]$, the signature chores in $S_i$ and the edge chore $E_i$ are assigned to the agent $e_i$. The dummy chores $D_1,D_2,D_3$ are assigned to the dummy agents $d_1,d_2,d_3$, respectively. The vertex chores are partitioned between the three dummy agents according to the coloring. That is, if the colors are red, blue, and green, then all chores corresponding to the red-colored vertices are assigned to $d_1$, while those corresponding to the green and blue-colored vertices are assigned to $d_2$ and $d_3$, respectively.

Observe that the allocation satisfies the aforementioned sufficient condition for rank-maximality. Therefore, we only need to establish envy-freeness. 

Recall that an allocation of chores is envy-free if and only if each agent considers the worst chore in its bundle to be more preferable that the worst chore in every other agent's bundle. In the allocation constructed above, for any $i \in [r]$, the worst chore in the bundle of agent $e_i$ (according to $e_i$'s importance order) is the edge chore $E_i$. For any other $j \neq i$, the worst chore in the bundle of $e_j$ (according to $e_i$'s importance order) is $E_j$, which is less preferable than $E_i$ to agent $e_i$. Furthermore, due to the rainbow coloring condition, the worst chore in the bundle of any dummy agent (according to $e_i$) is one of the vertex chores in $V_{e_i}$, which is again less preferable than $E_i$. Therefore, $e_i$ does not envy any other agent.

Let us now consider the perspective of the dummy agent $d_\ell$ for any fixed $\ell \in [3]$. According to $d_\ell$'s preferences, the worst chore in its own bundle is $D_\ell$, whereas that in the bundle of any other dummy agent is in $D \setminus D_\ell$, which is strictly less preferred. Similarly, the worst chore in the bundle of any edge agent $e_j$ is $E_j$ which, again, is less preferred. Therefore, the dummy agent does not envy any other agent either, implying that the allocation is envy-free.

($\Leftarrow$) Now suppose that there exists an envy-free and rank-maximal allocation, say $A$. By the necessary condition for rank-maximality, we know that for every $i \in [r]$, the edge chore $E_i$ is assigned to the edge agent $e_i$ (and this is the least preferred chore of $e_i$ in its bundle). Furthermore, all vertex chores in $V$ are allocated among the dummy agents. In order for agent $e_i$ to not envy the dummy agents, the vertex chores in $V_{e_i}$ must be allocated among the three dummy agents such that each dummy agent gets at least one chore in $V_{e_i}$. One can now infer a coloring of the hypergraph $H$, wherein the color of a vertex is uniquely determined by the index of the dummy agent that owns the corresponding vertex chore. It is easy to verify that this coloring satisfies the rainbow condition. This completes the proof of \Cref{thm:EF_RM_NP-complete_Chores}.
\end{proof}

\subsection{Envy-Freeness Up to Any Chore}

Next, we will consider a relaxation of envy-freeness called envy-freeness up to any chore (\EFX{}), which requires that any pairwise envy can be eliminated by removing any chore from the envious agent's bundle. Analogous to the goods setting~\cite{HSV+21fair}, one can characterize \EFX{} allocations of chores as ones where any envious agent gets exactly one chore~(\Cref{prop:efx_property_chores}).

\begin{restatable}[\textbf{\EFX{} characterization for chores}]{prop}{EFXPropertyChores}
An allocation $A$ of chores is \EFX{} if and only if each envious agent in $A$ gets exactly one chore.
\label{prop:efx_property_chores}
\end{restatable}

\begin{proof}
Consider an \EFX{} allocation $A$. Suppose for contradiction that there exists an envious agent $i$ such that $|A_i| \neq 1$. If $A_i = \emptyset$, then agent $i$ does not envy any other agent $h$.
Assume $|A_i| > 1$ and consider an identical preference profile.
By assumption, there exists an agent $h$ such that $A_{h} \succ_{i} A_{i}$. Since preferences are identical, then $i$'s envy toward $h$ cannot be eliminated by removing any of the chores $o \in A_i$, violating \EFX{} condition.

For the other direction, suppose $A$ is an allocation where every envious agent receives exactly one chore. Then by the definition of \EFX{} for chores, each envious agent can remove its single chore and eliminate any potential envy, implying that $A$ is \EFX{}.
\end{proof}

The characterization in \Cref{prop:efx_property_chores} can be used to show that a simple algorithm always returns an \EFX{} and Pareto optimal allocation of chores: Fix a priority ordering $\sigma$ over agents. Let the first agent in $\sigma$ pick its most preferred $m -n$ chores. Then, all agents (including the first agent) pick one chore each according to $\sigma$ from the remaining items.

\begin{restatable}[\textbf{\EFX{} + \PO{} for chores}]{prop}{EFXChores}
Given a chores instance, an \EFX{} and Pareto optimal allocation always exists and can be computed in polynomial time.
\label{prop:EFX+PO_Chores}
\end{restatable}

We show that the family of algorithms, specified by a priority ordering over agents $\sigma$ (Algorithm \ref{alg:EFXPO_Chores}) always returns an \EFX{} + \PO{} allocation for allocating chores.

\begin{algorithm}[t]
\DontPrintSemicolon
 \linespread{1.2}
\KwIn{A lexicographic chores instance $\langle N,M,\id \rangle$}
\KwOut{An \EFX{}+\PO{} allocation $A$}
\Parameters{A permutation $\sigma: N \rightarrow N$ of the agents}
$A \leftarrow (\emptyset,\dots,\emptyset)$\;
\If{$m > n$}{
Assign to $\sigma_1$, its most preferred $m-n$ chores according to $\id_{\sigma_1}$\;
}
Execute one round of serial dictatorship according to $\sigma$ on the remaining chores (a picking sequence with $\langle 1,1,\ldots, 1\rangle$).\;
\KwRet{$A$}
\caption {Algorithm for \EFX{}+\PO{} allocation of chores}
\label{alg:EFXPO_Chores}
\end{algorithm}

\begin{proof}
We start by showing that the allocation $A$ returned by Algorithm \ref{alg:EFXPO_Chores} is \EFX{}.
By \cref{prop:efx_property_chores}, for all agents who receive exactly one chore, i.e., $i\in [2, n]$ envy can be eliminated by removing that chore. Thus, we can focus on the case where $m\geq n$.
In this case, agent $1$ receives its most preferred $m- n+1$ chores, that is, $M\setminus\id_1([n-1], M)$. 
Since preferences are lexicographic, for any $i\in [2, n]$,  $\forall c\in A_1, c\succ_1 A_{i}$. Therefore, agent $1$ does not envy any other agent's chore. 
The rest of agents only receive one chore, and thus, by Proposition~\ref{prop:efx_property_chores} their potential envy can be eliminated by removing that chore, implying that the allocation is \EFX{}.

The proof of Pareto optimality is based on sequenciblity induced by Algorithm \ref{alg:EFXPO_Chores}. By \cref{prop:po_seqencible_chores}, an allocation of chores $A$ is \PO{} if and only if it is sequencible. Therefore, it suffices to show that Algorithm \ref{alg:EFXPO_Chores} can be implemented as a picking sequence, and therefore, it is \PO{}. By construction, Algorithm \ref{alg:EFXPO_Chores} is specified by a picking sequence $\sigma$, and therefore, its outcome is guaranteed to be \PO{} under lexicographic preferences.
\end{proof}

\Cref{prop:EFX+PO_Chores} is encouraging since, under additive valuations, an \EFX{} and \PO{} allocation is not guaranteed to exist; this follows from a straightforward adaption of a counterexample of \citeauthor{PR20almost} from the goods-only setting.\footnote{Indeed, consider an instance with two agents and three chores. The valuation profiles of the two agents are $(-2,-1,0)$ and $(-2,0,-1)$, respectively. Pareto optimality requires that the second chore be assigned to the second agent, while the third chore should be given to the first agent. Now, whichever agent gets the first chore will envy the other agent even after removing the zero-valued chore from its own bundle.} It is relevant to note that the goods-only counterexample of \citeauthor{PR20almost} as well as its adaptation to chores crucially use \emph{zero marginal} valuations, e.g. goods with zero values for some agents who have additive preferences. If, however, one insists on additive valuations with \emph{non-zero} marginals (notice that lexicographic preferences are a special case of this class), then, to our knowledge, showing the non-existence of an \EFX{}+\PO{} allocation remains an open problem for both the goods-only and chores-only problems. Thus, our result in \Cref{prop:EFX+PO_Chores} can be seen as progress towards this open question.

The positive result for lexicographic preferences in \Cref{prop:EFX+PO_Chores} motivates us to strengthen the efficiency notion to rank-maximality (\RM{}). Unfortunately, an \EFX{}+\RM{} allocation could fail to exist. Further, the intractability associated with exact envy-free and rank-maximal allocations~(\Cref{thm:EF_RM_NP-complete_Chores}) persists even when the fairness requirement is relaxed to \EFX{}.

\begin{restatable}[\textbf{\EFX{}+\RM{} for chores}]{thm}{EFXRMChores}
Determining whether a chores instance admits an \EFX{} and rank-maximal allocation is \NPC{}.
\label{thm:EFX_RM_NP-complete_Chores}
\end{restatable}

\begin{proof}
Membership in \NP{} follows from the fact that both \EFX{} and rank-maximality can be checked in polynomial time. To prove \NPH{}ness, we will follow the proof technique from the goods case (Theorem~3 in \citeauthor{HSV+21fair}, and show a reduction from a restricted version of \ThreeSAT{} called \TTTSAT{}, which is known to be \NPC{}~\cite{AD19sat}. An instance of \TTTSAT{} consists of a collection of $r$ variables $X_1,\dots,X_r$ and $s$ clauses $C_1,\dots,C_s$, where each clause is specified as a disjunction of three literals, and each variable occurs in exactly four clauses, twice negated and twice non-negated. The goal is to determine if there is a truth assignment that satisfies all clauses.

\emph{Construction of the reduced instance}: We will construct a fair division instance with $n = 2r+s$ agents and $m=3r+3s$ chores. The set of agents consists of $2r$ \emph{literal} agents $\{x_i,\overline{x}_i\}_{i \in [r]}$ and $s$ dummy agents $\{d_i\}_{i \in [s]}$. The set of chores consists of $2r$ \emph{signature} chores $\{S_i,\overline{S}_i\}_{i \in [r]}$, $s$ \emph{clause} chores $\{C_j\}_{j \in [s]}$, $2s$ \emph{dummy} chores $\{D_i,\overline{D}_i\}_{i \in [s]}$, and $r$ \emph{top} chores $\{T_i\}_{i \in [r]}$. 

\begin{table}[ht]
\footnotesize
\renewcommand{\arraystretch}{1.35}
\centering
    \begin{tabular}{|rl|}
    \hline
    $\vartriangleright^\dagger$: & $\overline{D}_s \id \dots \id \overline{D}_1 \id C_s \id \dots \id C_1 \id T_1 \id \overline{S}_r \id S_r \id \dots \id \overline{S}_1 \id S_1 \id D_s \id \dots \id D_1$\\
    \hline
    $x_i$: & $* \id \overline{S}_i \id S_i \id \vartriangleright^\dagger_{(s-k)} \id C_k \id \vartriangleright^\dagger_{(k-j-1)} \id C_j \id \vartriangleright^\dagger_{(j-1)} \id T_i$\\
    \hline
    $\overline{x}_i$: & $* \id \, \overline{S}_i \id S_i \id \vartriangleright^\dagger_{(s-q)} \, \id C_q \id \, \vartriangleright^\dagger_{(q-p-1)} \, \id C_p \id \, \vartriangleright^\dagger_{(p-1)} \, \id T_i$\\
    \hline
    $d_\ell$: & $* \, \id \overline{D}_\ell \, \id D_\ell$ \\
    \hline
    \end{tabular}
    \caption{Importance orderings of agents in the proof of \Cref{thm:EFX_RM_NP-complete_Chores}.}
    \label{tab:EFX_RM_NP-complete_chores}
\end{table}

\emph{Preferences}: \Cref{tab:EFX_RM_NP-complete_chores} shows the importance orderings of the agents. Let $\vartriangleright^\dagger$ define a \emph{reference} ordering on the set of chores. For every $i \in [r]$, if $C_{j}$ and $C_{k}$ denote the two clauses containing the positive literal $x_i$, then the literal agent $x_i$ ranks $T_i$ at the bottom, and the clause chores $C_j$ and $C_k$ at ranks $m-(j+1)$ and $m-(k+1)$, respectively. The missing positions consist of the remaining chores ranked according to $\vartriangleright^\dagger$ (we write $\vartriangleright^\dagger_{\ell}$ to denote the top $\ell$ chores in $\vartriangleright^\dagger$ that have not been ranked so far). The symbol $*$ indicates rest of the chores ordered according to $\vartriangleright^\dagger$. The preferences of the (negative) literal agent $\overline{x}_i$ and the dummy agents $d_1,\dots,d_s$ are defined similarly as shown in \Cref{tab:EFX_RM_NP-complete_chores}. This completes the construction of the reduced instance.

Note that for any fixed $i \in [r]$, the bottom (favorite) chore $T_i$ is ranked at the bottom position by the literal agents $x_i$ and $\overline{x}_i$, and at a higher position by all other agents. Therefore, any rank-maximal allocation must assign $T_i$ to either $x_i$ or $\overline{x}_i$. 
Similarly, the clause chore $C_j$ must be assigned to a literal agent corresponding to a literal contained in the clause $C_j$ in any rank-maximal allocation. Also note that for a similar reason, the signature chores $S_i,\overline{S}_i$ must be allocated between $x_i$ and $\overline{x}_i$, and, for any fixed $\ell \in [s]$, the dummy chores $D_\ell$ and $\overline{D}_\ell$ must be assigned to the dummy agent $d_\ell$ in any rank-maximal allocation.

The aforementioned \emph{necessary} conditions for rank-maximality are also \emph{sufficient} since each clause chore $C_j$ is ranked at the same position by all literal agents corresponding to the literals contained in clause $C_j$. A similar reasoning applies to the top ranked chores, the signature chores, and the dummy chores.

We will now argue the equivalence of solutions.

($\Rightarrow$) Given a satisfying truth assignment, the desired allocation, say $A$, can be constructed as follows: For any $i \in [r]$, if $x_i = 0$, then assign $\overline{S}_i$ to $x_i$, and assign $T_i$ and $S_i$ to $\overline{x}_i$. Otherwise, if $x_i = 1$, then assign $\overline{S}_i$ to $\overline{x}_i$, and assign $T_i$ and $S_i$ to $x_i$. Next, for every $\ell \in [s]$, assign the dummy chores $D_\ell,\overline{D}_\ell$ to the dummy agent $d_\ell$. Finally, each clause chore $C_j$ to assigned to a literal agent $x_i$ (or $\overline{x}_i$) if the literal $x_i$ (or $\overline{x}_i$) is contained in the clause $C_j$ and the clause is satisfied by the literal, i.e., $x_i = 1$ (or $\overline{x}_i = 1$). Note that under a satisfying assignment, each clause must have at least one such literal.

It is easy to verify that allocation $A$ satisfies the aforementioned sufficient condition for rank-maximality. To see why $A$ is \EFX{}, observe that for every $\ell \in [s]$, the dummy agent $d_\ell$ does not envy any other agent. Additionally, for any $i \in [r]$, the signature chores $S_i$ and $\overline{S}_i$ are the least-preferred chores of $x_i$ or $\overline{x}_i$ that they receive under $A$. As a result, the literal agents do not envy the dummy agents. Furthermore, if $x_i$ receives the signature chore $\overline{S}_i$, then it is not assigned any other chore, and therefore, its envy towards $\overline{x}_i$ can be eliminated by the removal of $\overline{S}_i$, implying that $A$ is \EFX{}.

($\Leftarrow$) Now suppose there exists an \EFX{} and rank-maximal allocation $A$. Then, $A$ must satisfy the aforementioned necessary condition for rank-maximality. That is, for every $\ell \in [s]$, the dummy chores $D_\ell,\overline{D}_\ell$ are assigned to the dummy agent $d_\ell$. In addition, for every $i \in [r]$, the top chore $T_i$ and the signature chores $S_i,\overline{S}_i$ are assigned between $x_i$ and $\overline{x}_i$. Furthermore, each clause chore $C_j$ is assigned to a literal agent $x_i$ (or $\overline{x}_i$) such that the literal $x_i$ (or $\overline{x}_i$) is contained in the clause $C_j$. 

Since $A$ is \EFX{}, we have from \Cref{prop:efx_property_chores} that if $x_i$ (or $\overline{x}_i$) is assigned the signature chore $\overline{S}_i$, then it is not assigned any other chore. Thus, in particular, the top chore $T_i$ and the other signature chore $S_i$ must be assigned to the other literal agent.

We will construct a truth assignment for the \TTTSAT{} instance as follows: For every $i \in [r]$, if $\overline{S}_i \in A_{x_i}$, then set $x_i = 0$, otherwise set $x_i = 1$. Note that the assignment is feasible as no literal is assigned conflicting values. To see why this is a satisfying assignment, consider any clause $C_j$. Suppose the clause chore $C_j$ is assigned to a literal agent $x_i$ (an analogous argument works when $\overline{x}_i$ gets $C_j$). Then, due to rank-maximality, we know that the literal $x_i$ must be contained in the clause $C_j$. Furthermore, from the aforementioned observation, it follows that $\overline{S}_i \in A_{\overline{x}_i}$, implying that $x_i = 1$. Thus, the clause $C_j$ is satisfied, as desired.
\end{proof}

\subsection{Envy-Freeness Up to One Chore}

In the search for positive results alongside rank-maximality, we further relax the \EFX{} notion to envy-freeness up to one chore (\EF{1}). However, we show that the intractability persists even under this relaxation.

\begin{restatable}[\textbf{\EF{1}+\RM{} for chores}]{thm}{EFoneRMChores}
Determining whether a given chores instance admits an \EF{1} and rank-maximal allocation is \NPC{}.
\label{thm:EF1_RM_NP-complete_Chores}
\end{restatable}

\begin{proof}
The proof is similar to that of \Cref{thm:EF_RM_NP-complete_Chores} but uses a slightly different construction.

We will once again show a reduction from \HRTC{}~\cite{DRS05hardness,GS17hardness}. Recall that under this problem, we are given a hypergraph $H = (V,E)$ with vertex set $V$ and edge set $E \subseteq 2^V$, and the goal is to determine if there exists a way of assigning each vertex exactly one of three given colors such that each hyperedge receives all three colors.

\emph{Construction of the reduced instance}:
Let $q \coloneqq |V|$ and $r \coloneqq |E|$ denote the number of vertices and hyperedges in the given hypergraph $H = (V,E)$, respectively. Let $\Delta \coloneqq \max_{e \in E} |V \cap e|$ denote the maximum number of vertices in any hyperedge in $E$.

We will construct a fair division instance with $n = r + 3$ agents and $m = q + (\Delta+2) \cdot r + 3$ chores. The set of agents consists of $r$ \emph{edge} agents $e_1,\dots,e_r$, and three \emph{dummy} agents $d_1,d_2,d_3$. The set of chores consists of $\Delta \cdot r$ \emph{Type I signature} chores $\{S_i^1,\dots,S_i^\Delta\}_{i \in [r]}$, $r$ \emph{Type II signature} chores $\{S'_1,\dots,S'_r\}$, $r$ \emph{edge} chores $E_1,\dots,E_r$, $q$ \emph{vertex} chores $V_1,\dots,V_q$, and three \emph{dummy} chores $D_1,D_2,D_3$.

With some notational overloading, we will write $E \coloneqq \{E_1,\dots,E_r\}$, $V \coloneqq \{V_1,\dots,V_q\}$, $S_i \coloneqq \{S_i^1,\dots,S_i^\Delta\}$, and $D \coloneqq \{D_1,D_2,D_3\}$ to denote the unordered sets of edge, vertex, Type I signature, and dummy chores, respectively. Further, we will write $S_{-i} \coloneqq \{S_1,S_2,\dots,S_{i-1},S_{i+1},\dots,S_r\}$ to denote the unordered set of Type I signature chores after removing those in $S_i$, $S'_{-i} \coloneqq \{S'_1,\dots,S'_{i-1},S'_{i+1},\dots,S'_r\} \setminus S'_i$ to denote the unordered set of Type II signature chores after removing the chore $S'_i$, and write $V_{e_i}$ to denote the unordered set of vertex chores corresponding to the vertices in the hyperedge $e_i \in E$. Note that for all $i \in [r]$, $|V_{e_i}| \leq q$ and $|S_i| = \Delta$.

\begin{table}[ht]
\footnotesize
\renewcommand{\arraystretch}{1.35}
\centering
    \begin{tabular}{|rl|}
    \hline
    $e_i$: & $S'_{-i} \id S_{-i} \id E \setminus \{E_i\} \id E_i \id V_{e_i}\id S'_i \id V \setminus V_{e_i} \id D \id S_i$\\
    \hline
    $d_\ell$: & $E_r \id \dots \id E_1 \id S'_{r} \id \dots \id S'_1 \id S_r \id \dots \id S_1 \id D \setminus D_\ell \id V \id D_\ell$\\
    \hline
    \end{tabular}
    \caption{Importance orderings of agents in the proof of \Cref{thm:EF1_RM_NP-complete_Chores}.}
    \label{tab:EF1_RM_NP-complete_Chores}
\end{table}

\emph{Preferences}: \Cref{tab:EF1_RM_NP-complete_Chores} shows the importance orderings of the agents in terms of a partial order over sets of items. To obtain the exact ordering, the partial orders can be extended by ranking the vertex chores in order of their indices (i.e., $V_q \id V_{q-1} \id \dots \id V_1$) and ranking the edge, signature (Type I and Type II), and dummy chores arbitrarily.

Notice that for every $i \in [r]$, any Type I signature chore in the set $S_i$, as well as the Type II signature chore $S'_i$, realizes its largest ranking in the preference list of agent $e_i$. Therefore, any rank-maximal allocation must assign all chores in $S_i$ as well as the chore $S'_i$ to $e_i$. By a similar reasoning, the edge chore $E_i$ is also assigned to $e_i$ in any rank-maximal allocation. Next, note that for any fixed $i \in [q]$, the vertex chore $V_i$ is ranked lower by the dummy agents than by the edge agents, and therefore must be assigned to one of $d_1$, $d_2$, or $d_3$ in any rank-maximal allocation. Finally, rank-maximality requires that for any $\ell \in [3]$, the dummy chore $D_\ell$ is assigned to the dummy agent $d_\ell$. It is straightforward to check that the aforementioned necessary conditions for rank-maximality are also sufficient.

We will now argue the equivalence of solutions.

($\Rightarrow$) Suppose the hypergraph $H$ admits a feasible rainbow $3$-coloring. Then, the desired allocation can be constructed just like in the proof of \Cref{thm:EF_RM_NP-complete_Chores}, as follows: For every $i \in [r]$, the Type I signature chores in $S_i$, the Type II signature chore $S'_i$, and the edge chore $E_i$ are assigned to the agent $e_i$. The dummy chores $D_1,D_2,D_3$ are assigned to the dummy agents $d_1,d_2,d_3$, respectively. The vertex chores are partitioned between the three dummy agents according to the coloring. That is, if the colors are red, blue, and green, then all chores corresponding to the red-colored vertices are assigned to $d_1$, while those corresponding to the green and blue-colored vertices are assigned to $d_2$ and $d_3$, respectively.

Rank-maximality is easy to check, so we will proceed to showing envy-freeness up to one chore (\EF{1}). Notice that for any $i \in [r]$, the worst chore in the bundle of agent $e_i$ (according to $e_i$'s importance order) is the edge chore $E_i$. For any other $j \neq i$, the worst chore in the bundle of $e_j$ (according to $e_i$'s importance order) is $E_j$, which is less preferable than $E_i$ to agent $e_i$. Therefore $e_i$ does not envy $e_j$. 

Next, note that due to the rainbow coloring condition, the chores in $V_{e_i}$ are assigned such that each dummy agent gets at least one chore in $V_{e_i}$. Therefore, the worst chore in the bundle of any dummy agent (according to $e_i$) is one of the vertex chores in $V_{e_i}$, implying that $e_i$ certainly envies each dummy agent. However, if the chore $E_i$ is removed from $e_i$'s bundle, then the worst chore in its residual bundle is the Type II signature chore $S'_i$, which it prefers over the chores in $V_{e_i}$. Thus, the allocation is \EF{1} as far as the envy from the edge agents to the dummy agents is concerned.

Let us now consider the perspective of the dummy agent $d_\ell$ for any fixed $\ell \in [3]$. According to $d_\ell$'s importance ordering, the worst chore in its own bundle is one of the vertex chores, whereas that in the bundle of any other dummy agent (according to $d_\ell$) is in $D \setminus D_\ell$, which is strictly less preferred. Similarly, the worst chore in the bundle of any edge agent $e_j$ (according to $d_\ell$) is $E_j$ which, again, is less preferred. Therefore, the dummy agent does not envy any other agent either, implying that overall the allocation is \EF{1}.

($\Leftarrow$) Now suppose that there exists an \EF{1} and rank-maximal allocation, say $A$. By the necessary condition for rank-maximality, we know that for every $i \in [r]$, the edge chore $E_i$ is assigned to the edge agent $e_i$ (and this is the least preferred chore of $e_i$ in its bundle). While $e_i$ does not envy any other edge agent, it certainly envies the dummy agents who only receive the dummy and the vertex chores. In order for $A$ to be \EF{1}, it must be that upon the removal of $E_i$, $e_i$'s envy towards each of the dummy agents should be eliminated. This, in particular, requires that each dummy agent must be assigned at least one of the chores in $V_{e_i}$. A rainbow coloring of the hypergraph $H$ can be easily inferred from this assignment. This completes the proof of \Cref{thm:EF1_RM_NP-complete_Chores}.
\end{proof}

\subsection{Maximin Share for Chores}

The intractability results in \Cref{thm:EFX_RM_NP-complete_Chores,thm:EF1_RM_NP-complete_Chores} prompt us to consider a different relaxation of \EFX{}, namely, maximin share (\MMS{}).
Before discussing the result for \MMS{} and rank-maximality, let us take a closer look at the structure of \MMS{} allocations under lexicographic preferences.

For lexicographic preferences, the \MMS{} partition of an agent for a chores-only instance is uniquely defined and can be efficiently computed. By contrast, for additive valuations, there can be multiple \MMS{} partitions, and unless $\textup{P=NP}$, there cannot be a polynomial-time algorithm for computing such a partition (since that would constitute an efficient procedure for solving the PARTITION problem). Given a chores instance and an agent $i\in N$ with importance ordering $\id_i$, agent $i$'s \MMS{} partition is uniquely given by $\{\{\id_{i}(1)\}, \{\id_{i}(2)\}, \ldots, \{\id_{i}(n-1)\}, \{\id_{i}(n), \ldots, \id_{i}(m)\} \}$.

\begin{restatable}[\textbf{\MMS{} characterization for chores}]{prop}{MMSPropertyChores}
An allocation of chores is \MMS{} if and only if any agent who receives its worst chore does not receive another chore.
\label{prop:mms_property_chores}
\end{restatable}

\begin{proof}
Consider the agent $i \in N$. Let $\id_i = o_1^- \id o_2^- \id \dots \id o_m^-$. Then, the \MMS{} partition of agent $i$ is uniquely defined, and its \MMS{} value is given by  $\MMS_i = \min \left\{\{o_1\}, \{o_{2}\}, \ldots, \{o_{n-1}\},  \{o_n,\dots,o_{m}\}\right\}$.
Since preferences are lexicographic, any bundle $A_i$ that does not contain $o_1$ is preferred to $A'_i = \{o_1\}$. Moreover, $A'_{i}$ satisfies \MMS{} by definition. 
\end{proof}

\begin{restatable}[\textbf{\EFX{} vs \MMS{}}]{prop}{EFXimpliesMMS}
Given any chores instance with lexicographic preferences, any \EFX{} allocation satisfies \MMS{} but the converse is not always true.
\label{prop:EFX_implies_MMS}
\end{restatable}

\begin{proof}
Suppose for contradiction that there exists an \EFX{} allocation $A$ that is not \MMS{}. Then, by \cref{prop:efx_property_chores} no envious agent receives more than one chore in $A$. Since $A$ is not \MMS{}, then there must exist an agent, say $i$, that receives its least preferred chore (first chore in its importance ordering) along with some other chore. When all agents have identical preferences, agent $i$ must be envious of some other agent, say the agent that receives the best chore. 
However, this contradicts \cref{prop:efx_property_chores} since an envious agent (agent $i$) has received more than one chore.

To prove that \MMS{} does not always imply \EFX{}, consider the following instance with four agents and five chores:
\begin{align*}
    1: {o_5^-} \, \id \, o_4^- \, \id \, o_3^- \, \id \, o_2^- \, \id \, \underline{o_1^-}\\ \nonumber
    2: o_5^- \, \id \, {o_4^-} \, \id \, \underline{o_3^-} \, \id \, \underline{o_2^-} \, \id \, o_1^-\\ \nonumber
    3: {o_5^-} \, \id \, o_3^- \, \id \, o_2^- \, \id \, o_1^- \, \id \, \underline{o_4^-}\\ \nonumber
    4: {o_4^-} \, \id \, o_3^- \, \id \, o_2^- \, \id \, o_1^- \, \id \, \underline{o_5^-} \nonumber
\end{align*}
The allocation $\{\{o_1\},\{o_2,o_3\},\{o_4\},\{o_5\}\}$ satisfies \MMS{} since all agents receive a bundle excluding their worst chore. However, by \cref{prop:efx_property_chores} this allocation is not $\EFX{}$ because an envious agent (agent $2$) receives more than one chore.
\end{proof}

Note that \EF{1} and \MMS{} are incomparable even under lexicographic preferences.

\begin{example}[\textbf{\MMS{} vs \EF{1}}] \label{rem:EF1vMMS}
It can also be shown that \EF{1} and \MMS{} are incomparable notions in that one does not always imply the other. The fact that \MMS{} does not imply \EF{1} follows from the example in the proof of \Cref{prop:EFX_implies_MMS}: Indeed, agent $2$ continues to envy agent $1$ even after the removal of any good from the latter's bundle. To prove that \EF{1} does not imply \MMS{}, consider the following instance with identical preferences:
\begin{align*}
    1: \underline{o_5^-} \, \id \, o_4^- \, \id \, o_3^- \, \id \, o_2^- \, \id \, \underline{o_1^-}\\ \nonumber
    2: {o_5^-} \, \id \, o_4^- \, \id \, \underline{o_3^-} \, \id \, o_2^- \, \id \, {o_1^-}\\ \nonumber
    3: {o_5^-} \, \id \, o_4^- \, \id \, o_3^- \, \id \, \underline{o_2^-} \, \id \, {o_1^-}\\ \nonumber
    4: {o_5^-} \, \id \, \underline{o_4^-} \, \id \, o_3^- \, \id \, o_2^- \, \id \, {o_1^-}\\ \nonumber
\end{align*}
The underlined allocation $\{\{o_1, o_5\},\{o_3\},\{o_2\},\{o_4\}\}$ is \EF{1}; in particular, agent $1$'s envy towards all other agents can be eliminated by removing chore $o_5$ from its bundle.
However, this allocation fails to satisfy \MMS{} since agent $1$ receives a strict superset of its least preferred chore, which contradicts \Cref{prop:mms_property_chores}.
\end{example}

\Cref{prop:EFX_implies_MMS} and \Cref{prop:EFX+PO_Chores} together give the following implication.

\begin{restatable}[\textbf{\MMS{}+\PO{} for chores}]{cor}{MMSChores}
Given a chores instance, an \MMS{} and Pareto optimal allocation always exists and can be computed in polynomial time.
\label{prop:MMS+PO_Chores}
\end{restatable}

Note that an \MMS{} (or \MMS{}+\PO{}) allocation could fail to exist under additive valuations~\cite{KPW18fair}. \Cref{prop:EFX+PO_Chores} provides a characterization of an agent's \MMS{} partition which enables us to compute an \MMS{}+\PO{} allocation. However, \MMS{} may no longer exist when rank-maximality is additionally required, even under lexicographic preferences~(\Cref{eg:MMS_RM_NonExistence}).
Nevertheless, an \MMS{} and rank-maximal allocation can be computed efficiently whenever such an allocation exist~(\Cref{thm:MMS_RM_Polytime_Chores}).

\begin{example}[\textbf{\MMS{}+\RM{} may not exist}]
Consider two agents with the following preferences:
\begin{align*}
    1: {o}^-_1 \, \id \, {o}^-_2 \, \id \, {o}^-_3 \\ \nonumber
    2: {o}^-_1 \, \id \, {o}^-_3 \, \id \, {o}^-_2 \nonumber
\end{align*}
A rank-maximal allocation will assign $o_2$ to $2$ and $o_3$ to $1$, but then \MMS{} is violated for whichever agent gets $o_1$.
\label{eg:MMS_RM_NonExistence}
\end{example}

\begin{restatable}[\textbf{\MMS{}+\RM{} for chores}]{thm}{MMSRMChores}
There is a polynomial-time algorithm that, given as input a lexicographic chores instance, computes an \MMS{} and rank-maximal allocation, whenever one exists.
\label{thm:MMS_RM_Polytime_Chores}
\end{restatable}

\begin{proof}
Fix any agent $i \in N$, and suppose its preference is given by $\id_i \coloneqq o_1^- \, \id \, o_2^- \, \id \dots \id \, o_m^-$. Under lexicographic preferences, the \MMS{} partition of agent $i \in N$ is uniquely defined as
$$\left\{\{o_1\}, \{o_{2}\}, \ldots, \{o_{n-1}\},  \{o_n,\dots,o_{m}\}\right\}.$$

The definition of \MMS{} immediately gives a characterization for \MMS{} allocations: An allocation is \MMS{} if and only if no agent receives a bundle that is a strict superset of its least preferred chore (first chore in the importance ordering).

The key observation is that any arbitrary rank-maximal allocation that does not assign an agent its top ranked chore satisfies \MMS{}. Given this observation, we construct a polynomial-time algorithm to compute an \MMS{} and rank-maximal allocation, whenever one exists:
Compute a rank-maximal allocation $A$ given a preference profile $\id$. If no agent receives its top ranked chore under $A$ (chore in first position of the importance ordering, i.e., $\id_i(1)$), then by \cref{prop:mms_property_chores} allocation $A$ must also be \MMS{}, and we obtain an \MMS{} and \RM{} allocation.

On the other hand, let us assume that the rank-maximal allocation $A$ assigns a chore $o$ to an agent $i\in N$ such that for all $o'\in M$, $o \, \id_{i} \, o'$, that is, $o$ is the top-ranked chore for $i$. We show that  $o$ is a unique chore that is ranked first by every other agent, that is, for all $h\in N$ we have $\id_h(1) = o$.

By rank-maximality of $A$, it must be the case that all other agents also ranked chore $o$ as their top ranked chore, that is, for all agents $i'\in N$, for all chores $o'\in M$, $o \, \id_{i'}\, o'$.
To see why this is correct, suppose that there exists an agent $h$ who ranked $o$ lower in its importance ordering. Then, chore $o$ can be assigned to $j$. This improves the signature of the allocation, contradicting the rank-maximality of $A$. Recall that rank-maximality for the chores-only case is defined as an allocation that maximizes the number of agents who receive their lowest ranked chore, subject to that, maximizes the number of agents who receive their second lowest chore, and so on.
Therefore, we can repeatedly check for each agent that can potentially be assigned its top ranked chore, whether it receives any other chore under an $\RM{}$ allocation. 

Let $o$ be such chore that is ranked first by all agents.
In each iteration, fix an agent $i$ and assign chore $o$ to it. Remove agent $i$ and chore $o$ and compute a rank-maximal allocation, say $A'$, on the reduced instance.
Now compare the signature of $A$ with $A'$. If the first $m-1$ elements in the signature of $A$ and $A'$ are equal (i.e., signature associated with the best $m-1$ chores), then the instance admits an \MMS{} and rank-maximal allocation. Otherwise, go to the next iteration.
Return NO if all agents are checked above without finding an \MMS{} and rank-maximal allocation.
All steps above including computing a rank-maximal allocation can be done in polynomial time.
\end{proof}

\section{Mixed Items: Additional Results}
\label{sec:Mixed}

Recall that in the mixed-item setting, high-ranked goods in the importance ordering are desirable while high-ranked chores need to be avoided.

We have already seen in \cref{sec:Mixed} that an \EFX{} allocation may not always exist. Next, we introduce two relaxations of \EFX{} and discuss their existence and computation.

\subsection{\EFX{} Variations for Mixed Instances} \label{app:EFX_variation}

We define two relaxations of \EFX{} for mixed items. An allocation $A$ is
\begin{itemize}[wide,labelindent=0pt]
    \item {\bf Envy-free up to any good (\EFX{}-g)}, if for every pair of agents $i,h\in N$ such that $A_h^{i+} \neq \emptyset$, and for {\em every} good $g\in A_h^{i+}$, it holds that $A_i\succeq_i A_h\setminus\{g\}$, and
    \item {\bf Envy-free up to any chore (\EFX{}-c)}, if for every pair of agents $i,h\in N$ such that $A_i^{i-} \neq \emptyset$, for {\em every} chore $c\in A_i^{i-}$, it holds that $A_i\setminus\{c\}\succeq_i A_h$.
\end{itemize}
Notice that \EFX{} implies both \EFX{}-g and \EFX{}-c.

First, we show relaxing \EFX{} to \EFX{}-c is not sufficient to revive the non-existence result of \EFX{} (as shown in \cref{thm:EFX_counterexample}). In fact, the same counterexample introduced in the proof of \cref{thm:EFX_counterexample} can be used to illustrate that \EFX{}-c may not always exist even for objective instances. 

\begin{restatable}[\textbf{\EFX{}-c for mixed items}]{prop}{EFXcMixed}
There exists an instance with objective mixed items and lexicographic preferences that does not admit any \EFX{}-c allocation.
\label{prop:EFXc_Mixed}
\end{restatable}

\begin{remark}[\textbf{\EFX{}-c v. \EFX{}-g v. \MMS{}}]
Note that for mixed items, \EFX{}-c does not necessarily imply \EFX{}-g, \EFXg{} does not always imply \EFXc{}, and neither \EFXc{} nor \EFXg{} necessarily imply MMS.

Consider two agents with identical importance ordering as follows: $\id_1 = \id_2 =  \underline{o_1^{-}} \, \id \, \underline{o_2^{+}} \, \id \, o_3^{+} \, \id \, o_4^{+}$. If agent $1$ gets the underlined items (one chore and one good) and another agent gets the other two goods, then such an allocation is \EFX{}-c (up to the removal of chore $o_1^-$). However, this allocation is neither \EFX{}-g (removal of either good $o_3^+$ or $o_4^+$ does not eliminated $1$'s envy for $2$), nor \MMS{}.

Now, consider two agents with identical importance ordering as follows: $\id_1 = \id_2 =  \underline{o_1^{+}} \, \id \, \underline{o_2^{-}} \, \id \, o_3^{-} \, \id \, o_4^{-}$. Suppose agent $1$ gets the underlined items (one good and one chore) and agent $2$ gets the other two chores. Such an allocation is \EFX{}-g (up to the removal of good $o_1^+$), but is neither \EFX{}-c (removal of chore $o_3^-$ does not eliminated $2$'s envy for $1$), nor \MMS{} ($2$'s maximin share is $\emptyset$ by \Cref{prop:MMS_mixed_char}).
\end{remark}

\subsection{Envy-Freeness Up to One Item}

For objective mixed items, although an \EFX{} allocation may not exist, an \EF{1} allocation can be computed in polynomial time through the double round-robin algorithm~\cite{ACI+19fair}.
The algorithm proceeds in two phases. First, agents pick chores in a round robin fashion according to a fixed order. Second, the order is reversed, and agents pick goods by round-robin according to this reversed order.

Unfortunately, the double round-robin algorithm algorithm fails to guarantee \EFX{} even for two agents with objective goods and chores. 
Consider two identical agents with preferences: $o^+_1 \id {o}^-_2 \id o^+_3 \id {o}^-_4 \id o^+_5$. The double round-robin with priority ordering of $\sigma=(1,2)$ will assign $\{o^+_3, {o}^-_4\}$ to agent 1 and $\{o^+_1, o^+_5, {o}^-_2 \}$ to agent 2. This allocation is not \EFX{} because agent 1's envy cannot be eliminated by removing, say, the good $o^+_5$. 

The existence and computation of \EF{1} allocations along with \PO{} even for objective mixed items remains an open problem. Moreover, determining whether an instance admits an \EFX{} allocation (with or without \PO{}) remains open for (objective and subjective) mixed items.

\end{document}